\documentclass[a4paper,11pt]{amsart}

\usepackage[left=2.5cm,right=2.5cm,top=4cm,bottom=4cm]{geometry}

\usepackage{amsmath}
\usepackage{amssymb}
\usepackage{amsthm}

\usepackage{pgf,tikz}
\usetikzlibrary{shapes,automata,positioning,arrows,calc}
\usepackage{color}
\usepackage{graphicx}
\usepackage{setspace}
\usepackage[version=3]{mhchem}
\usepackage{chemfig} 
\usepackage[sort&compress,comma,square,numbers]{natbib}
\usepackage{booktabs}
\usepackage{enumitem,hyperref}
\usepackage{url}
\usepackage{mathabx}

\newcommand{\R}{\mathbb{R}}
\newcommand{\N}{{\mathbb N}}
\newcommand{\Z}{{\mathbb Z}}

\newcommand{\newton}[1]{\ensuremath{\operatorname{N}(#1)}}
\newcommand{\vertices}[1]{\ensuremath{\operatorname{Vert}(#1)}}

\DeclareMathOperator*{\sign}{sign}

\renewcommand{\k}{{\kappa}}
\newcommand{\st}{\mid}
\setlist[itemize]{leftmargin=*}
\setlist[enumerate]{leftmargin=*}

\newtheorem{theorem}{Theorem}
\numberwithin{theorem}{section}
\newtheorem{lemma}[theorem]{Lemma}
\newtheorem{proposition}[theorem]{Proposition}
\newtheorem{corollary}[theorem]{Corollary}
\theoremstyle{definition}
\newtheorem{definition}[theorem]{Definition}
\newtheorem{example}[theorem]{Example}
\newtheorem{remark}[theorem]{Remark}

\definecolor{NiceBlue}{rgb}{0.2,0.2,0.75}

\begin{document}

\title{The kinetic space of multistationarity in dual phosphorylation}

\author[E. Feliu, N. Kaihnsa, T. de Wolff, O. Y\"ur\"uk]{Elisenda Feliu$^{1,4}$, Nidhi Kaihnsa$^{2}$, Timo de Wolff$^3$, O\u{g}uzhan Y\"ur\"uk$^3$}



\date{\today}

\footnotetext[1]{Department of Mathematical Sciences, University of Copenhagen, Universitetsparken 5, 2100 Copenhagen, Denmark. efeliu@math.ku.dk}
 
 \footnotetext[2]{Division of Applied Mathematics, Brown University, 182 George Street, Providence, RI, 02912. nidhi\_kaihnsa@brown.edu}

\footnotetext[3]{Institute of Analysis and Algebra, TU Braunschweig, Universit\"atsplatz 2, 38106 Braunschweig, Germany. t.de-wolff@tu-braunschweig.de, oguyueru@tu-braunschweig.de}

\footnotetext[4]{\emph{Corresponding author: } efeliu@math.ku.dk}

 \tikzset{every node/.style={auto}}
 \tikzset{every state/.style={rectangle, minimum size=0pt, draw=none, font=\normalsize}}
 \tikzset{bend angle=15}

\begin{abstract}
Multistationarity in molecular systems underlies switch-like responses in cellular decision making. Determining whether and when a system displays multistationarity is in general a difficult problem. In this work we completely determine the set of kinetic parameters that enable multistationarity in a ubiquitous motif involved in cell signaling, namely a dual phosphorylation cycle. In addition we show that the regions of multistationarity and monostationarity are both path connected.

We model the dynamics of the concentrations of the proteins over time by means of a parametrized polynomial ordinary differential equation (ODE) system arising from the mass-action assumption. Since this system has three linear first integrals defined by the total amounts of the substrate and the two enzymes, we study for what parameter values the ODE system has at least two positive steady states after suitably choosing the total amounts. 
We employ a suite of techniques from (real) algebraic geometry, which in particular concern the study of the signs of a  multivariate polynomial over the positive orthant and sums of nonnegative circuit polynomials. 
\end{abstract}

\maketitle
 
 \section{\bf Introduction}
\emph{Multistationarity},  that is  the existence of multiple steady states in a system,  has been linked to cellular
decision making    and switch-like responses to graded input \cite{laurent1999,ozbudak2004,Xiong:2003jt}.    
In the context of chemical reaction networks, there exist numerous methods to decide whether multistationarity arises for some choice of parameter values \cite{Feinbergss,feliu_newinj,wiuf-feliu,PerezMillan,conradi-PNAS,craciun2008,control,crnttoolbox}. However, determining for \emph{which} parameter values this is the case, is a very difficult problem with complicated answers. 
Some recent progress in understanding the parameter region of multistationarity has eased the problem by focusing on subsets of parameters, and providing regions that guarantee or exclude that the other parameters can be chosen in such a way that multistationarity arises \cite{FeliuPlos,dickenstein:regions}. 

Here, we completely characterize the region of multistationarity in terms of kinetic parameters for a simple model of phosphorylation and dephosphorylation, which is a building block of the MAPK cascade involved ubiquitously in  cell signaling  \cite{Huang-Ferrell,qiao:oscillations,rendall-MAPK}. Phosphorylation processes are central in the modulation of cell communication, activities and responses, as, for example, phosphorylation affects about $30\%$ of all  proteins in human body \cite{cohen}.  

The reaction network we consider consists of a substrate $S$ that has two phosphorylation sites. Phosphorylation occurs distributively in an ordered manner, such that one of the sites is always phosphorylated first. We denote the three phosphoforms of $S$ with $0,1,2$ phosphorylated sites by $S_0,S_1,S_2$ respectively, and assume that a kinase $E$ and a phosphatase $F$ mediate the phosphorylation  and dephosphorylation of $S$ respectively.
 This gives rise to the following mechanism \cite{Wang:2008dc,conradi-mincheva}:
 \begin{align}\label{eq:network}
 \begin{split}
 S_0 + E \ce{<=>[\k_1][\k_2]} ES_0 \ce{->[\k_3]} S_1+E   \ce{<=>[\k_7][\k_8]} ES_1 \ce{->[\k_9]} S_2+E \\
 S_2 + F  \ce{<=>[\k_{10}][\k_{11}]} FS_2 \ce{->[\k_{12}]} S_1+F  \ce{<=>[\k_4][\k_5]} FS_1 \ce{->[\k_6]} S_0+F.
 \end{split}
 \end{align}
 Under the assumption of mass-action kinetics, the evolution of the concentration of the species of the network over time is modeled by a system of autonomous  ODEs in $\R^9_{\geq 0}$, see equation \eqref{eq:ode}. The system consists of polynomial equations, whose coefficients are scalar multiples of one of  $12$ positive parameters $\k_1,\dots,\k_{12}$. Furthermore, the dynamics are constrained to   linear invariant subspaces of dimension six, characterized  by the total amounts of kinase, phosphatase and substrate, which then enter the study as parameters.

In addition to the biological relevance of this system, this network  has become the \emph{model model} (like the model organisms in biology), where new techniques, strategies, and approaches are tested. We expect that the strategies employed to answer mathematical questions about this model can be used to approach similar systems arising in molecular biology.
This system is large enough for hands-on approaches to fail, but small enough to challenge the development of new mathematics. 
 Furthermore, dynamical properties of the ODE system of this network might be lifted to more complex networks related to it. For example, \eqref{eq:network} is an example of an $n$-site phosphorylation cycle \cite{Wang:2008dc,TG-Nature,FHC14}, a post-translational modification network \cite{TG-rational,fwptm,conradi-shiu-review}, a MESSI system \cite{Dickenstein-MESSI}, and a network with toric steady states \cite{PerezMillan}, to name a few.
 
 Currently, it is known 
 that the number of positive steady states within a linear invariant subspace is either one or three, if all positive steady states are  nondegenerate \cite{Wang:2008dc,Markevich-mapk}. It has also been shown that there are choices of parameters for which there are two asymptotically stable steady states and one unstable steady state \cite{rendall-2site}, see also \cite{torres:stability}. It is currently unknown whether it admits Hopf bifurcations or periodic solutions \cite{conradi-feliu-mincheva}. 
 
 Some recent progress has shed some light on how these qualitative properties depend on the choice of parameters. In \cite{conradi-mincheva} the authors give two rational functions $a(\k)$ and $b(\k)$ on the parameters  $\k_1,\dots,\k_{12}$ (see \eqref{eq:ab} below),  with the following properties: The system has one positive steady state in each invariant linear subspace if  $a(\k)\geq 0$ and $b(\k)\geq 0$, and has at least two in some invariant linear  subspace if $a(\k)<0$, see Subsection~\ref{sec:ineq}. Furthermore, in \cite{Feliu:royal,dickenstein:regions} conditions for the existence of three positive steady states involving the parameters  $\k_1,\dots,\k_{12}$ and some of the total amounts are given, see also \cite{conradi-erk}.

 The difficulties in understanding the number of steady states  arise from the high number of parameters and variables combined with the difficulties in studying polynomials over the positive real numbers. 
 This is what left  the scenario $a(\k)\geq 0$ and $b(\k)<0$ open in \cite{conradi-mincheva}. In this work, we focus on this open case. We give necessary conditions and sufficient conditions for multistationarity to arise in this case, and give an explicit parametrization of the boundary between the region of monostationarity and multistationarity. 
Specifically, our approach to the study of the regions of mono- and multistationarity gives rise to the following contributions:
\begin{itemize}
\item \textbf{Sufficient conditions for monostationarity.} We provide two such conditions of the form $H(\k)\geq 0$. First, we obtain a polynomial inequality in $\k$  using the theory of discriminants, see Theorem~\ref{prop:apos_disc} in Subsection~\ref{sec:pol}. This inequality completely characterizes the region of multistationarity when $a(\k)=0$. Second, we provide an inequality  where $H$ is a generalized polynomial with rational exponents. This is obtained by decomposing a relevant polynomial into a sum of nonnegative circuit polynomials (SONC), see Theorem~\ref{prop:apos_circuit} in Subsection~\ref{sec:circuit}. Although these inequalities are not necessary for monostationarity, the latter inequality gives preliminary information on the shape of the multistationarity region (Corollary~\ref{cor:apos_circuit}), which is critical to its characterization in Section~\ref{sec:enabled}.

\smallskip
\item \textbf{Existence of multistationarity when $a(\k)\geq 0$ and $b(\k)<0$. } Proposition~\ref{prop:multi:greenline} in Subsection~\ref{sec:multi:greenline} shows that in this case, multistationarity occurs for suitable $\k$.

\smallskip
\item \textbf{Parametric description of the regions of mono- and multistationarity. } In Theorem~\ref{thm:parametric} in Subsection~\ref{sec:param} we provide a full parametric description of the two regions, by giving an explicit parametric representation of the boundary between the two regions.

\smallskip
\item \textbf{Connectivity. } In Theorem~\ref{thm:connected} in Section~\ref{sec:connected} we conclude that the region of multistationarity in the parameters $\k_1,\dots,\k_{12}$ is an open and connected set, and the region of monostationarity is closed in $\R^{12}_{>0}$ and connected as well.
\end{itemize}

We will repeatedly employ the Descartes' rule of signs, and the study of the Newton polytope associated with several polynomials, the relevant properties of which are reviewed in Subsection~\ref{sec:newton}. 
Furthermore, some proofs rely on the use of symbolic algorithms from real algebraic geometry as implemented in {\tt Maple 2019}. 
These include the   selection of a point in each connected component of a semi-algebraic set, and the verification that a semi-algebraic set is empty. These computations are presented in the accompanying supplementary file {\it SupplInfo.mw}. Computations have also been performed in {\tt Mathematica}, to reassess the validity of the proofs. 

 We hope that the techniques used here, targeting the study of the signs of a parametric multivariate polynomial on the positive orthant, can be employed for other systems. For instance, the allosteric kinase model given in \cite{feng:allosteric} presents difficulties analogous to those encountered here.
Furthermore, the study of signs plays a key role when analyzing the stability of steady states or the presence of Hopf bifurcations via the Routh-Hurwitz criterion (see for example \cite{torres:stability,shiu:hopf}).

\section{\bf Preliminaries}\label{sec:prelim}
We start by introducing the notation, the ODE system and the mathematical techniques used in later sections, namely the Newton polytope and circuit polynomials. We elaborate on the problem we are interested in, and on the previous work. 

\subsection{The ODE system and a polynomial}\label{sec:system} 
We introduce the ODE system describing the dynamics of the reaction network \eqref{eq:network}, its linear first integrals, and a polynomial whose signs determine whether multiple positive steady states exist in some linear invariant subspace.

We consider the reaction network \eqref{eq:network} and 
denote the concentrations of the species by
$x_1=[E], x_2=[F]$, $x_3=[S_0]$, $x_4=[S_1]$, $x_5=[S_2]$, $x_6=[ES_0]$, $x_7=[FS_1]$, $x_8=[ES_1]$, $x_9=[FS_2]$. Under mass-action kinetics, the ODE system modelling the concentrations of the nine species in the network \eqref{eq:network}  over time $t$  is
{\small \begin{align}
\tfrac{dx_1}{dt} &= -\k_{1}x_{1}x_{3}-\k_{7}x_{1}x_{4}+\k_{2}x_{6}+\k_{3}x_{6}+\k_{8}x_{8}+\k_{9}x_{8} & \tfrac{dx_6}{dt} &= \k_{1}x_{1}x_{3}-\k_{2}x_{6}-\k_{3}x_{6} \nonumber
\\ 
\tfrac{dx_2}{dt} &= -\k_{4}x_{2}x_{4}-\k_{10}x_{2}x_{5}+\k_{5}x_{7}+\k_{6}x_{7}+\k_{11}x_{9}+\k_{12}x_{9} & \tfrac{dx_7}{dt} &= \k_{4}x_{2}x_{4}-\k_{5}x_{7}-\k_{6}x_{7} \nonumber
\\ 
\tfrac{dx_3}{dt} &=-\k_{1}x_{1}x_{3}+\k_{2}x_{6}+\k_{6}x_{7} & \tfrac{dx_8}{dt} &= \k_{7}x_{1}x_{4}-\k_{8}x_{8}-\k_{9}x_{8}  \label{eq:ode} \\
\tfrac{dx_4}{dt} &= -\k_{4}x_{2}x_{4}-\k_{7}x_{1}x_{4}+\k_{3}x_{6}+\k_{5}x_{7}+\k_{8}x_{8}+\k_{12}x_{9} & \tfrac{dx_9}{dt} &= \k_{10}x_{2}x_{5}-\k_{11}x_{9}-\k_{12}x_{9}   \nonumber
\\
\tfrac{dx_5}{dt} &=-\k_{10}x_{2}x_{5}+\k_{9}x_{8}+\k_{11}x_{9}, \nonumber
\end{align}}%
where $x_i=x_i(t)$, \cite{conradi-mincheva}.   This is a polynomial ODE system with coefficients $\k_1,\dots,\k_{12}>0$. These coefficients are treated as parameters, and referred to as \emph{reaction rate constants}. 
The positive and nonnegative orthants of $\R^9$ are forward invariant by the trajectories of this system (as it is the case for all mass-action systems \cite{volpert}). 
Furthermore,   the system admits exactly three independent linear first integrals, $x_1+ x_6 + x_8$, $x_2+x_7+x_9$ and $ x_3+x_4+x_5+x_6+x_7+x_8+x_9$. Note that these are independent of $\k_i$. 
It follows that the dynamics  take place in linear invariant subspaces of dimension six, defined by the equations
\begin{equation}\label{eq:cons_laws}
x_1+ x_6 + x_8 = E_{\rm tot},\quad x_2+x_7+x_9=F_{\rm tot},\quad  x_3+x_4+x_5+x_6+x_7+x_8+x_9= S_{\rm tot},
\end{equation}
subject to $x_i\geq 0$ for $i=1,\dots,9$. Here $E_{\rm tot}, F_{\rm tot}, S_{\rm tot}$ stand for the total amounts of kinase $E$, phosphatase $F$ and substrate $S$.
In the chemistry literature, the equations in \eqref{eq:cons_laws} are referred to as conservation laws and they define the so-called stoichiometric compatibility classes. 

\medskip
The steady states of the network are the solutions to the  system of polynomial equations given by setting the right-hand side of \eqref{eq:ode} to zero.  Three of these equations are redundant,  and for example the ones for $x_1,x_2,x_3$ can be removed.
The remaining six equations together with the equations in \eqref{eq:cons_laws} form the \emph{steady state system}, which  has variables $x_1,\dots,x_{9}$ and parameters $\k_1,\dots,\k_{12}$, $E_{\rm tot}, F_{\rm tot}, S_{\rm tot}$, all of which are assumed to be positive. The nonnegative solutions of the steady state equations determine the nonnegative steady states within the corresponding linear invariant subspace. This system has at least one positive solution for any choice of parameters, but it can have up to three. This gives rise to the following definition.

\begin{definition}\label{def:enable}
A vector of reaction rate constants  $\k=(\k_1,\dots,\k_{12})\in \R^{12}_{>0}$ \emph{enables} multistationarity if there exist $E_{\rm tot}, F_{\rm tot}, S_{\rm tot}$  such that the steady state system has at least two positive solutions, that is, with all coordinates positive. In this case we say that the network is multistationary in the linear invariant subspace with total amounts  $E_{\rm tot}, F_{\rm tot}, S_{\rm tot}$. 
The vector $\k$ is said to \emph{preclude} multistationarity, if it does not enable~it.
\end{definition}

In \cite{conradi-mincheva}, see also \cite{FeliuPlos}, sufficient conditions on the reaction rate constants for enabling or precluding multistationarity were given. These are reviewed in Subsection~\ref{sec:ineq}, after introducing a key polynomial and some background on signs of polynomials. 
Consider the Michaelis-Menten constants of each phosphorylation/dephosphorylation event:
\[K_1 = \tfrac{\k_2+\k_3}{\k_1},\quad  K_2 = \tfrac{\k_5+\k_6}{\k_4},\quad K_3 = \tfrac{\k_8+\k_9}{\k_7},\quad K_4 = \tfrac{\k_{11}+\k_{12}}{\k_{10}}.\]
The map
$\pi\colon \R^{12}_{>0}  \rightarrow  \R^{8}_{>0}$ sending $\k=(\k_1,\dots,\k_{12})$ to   $\eta=(K_1,K_2,K_3,K_4,\k_3,\k_6,\k_9,\k_{12})$
is continuous and surjective. 
Consider the following polynomial in $x_1,x_2,x_3$ with coefficients depending on~$\eta$:
	{\small \begin{align}\label{eq:mypolynomial}
	\begin{split}
	p_\eta(x)& =K_2\k_3 ( \k_{3}\k_{12}-\k_{6}\k_{9})  \Big( K_{2}K_{4}  \k_{3}\k_{9} x_{1}^{4}x_{3}^{2}
	 + K_{1}K_{3} \k_{6} \k_{12} (x_{1}^{3}x_{2}^{2}x_{3}  +  x_{1}^{2}x_{2}^{3}x_{3} 
+  x_{1}^{2}x_{2}^{2}x_{3}^{2}) 		 \\ &     +K_{2}K_{3} \k_{3}\k_{12}  x_{1}^{3}x_{2}x_{3}^{2} \Big)     
 +K_{1}K_{2}K_{3}\k_{3}\k_{6}\k_{12} ( (K_{2}+K_3)\k_{3}\k_{12}-(K_{1}+K_4)\k_{6}\k_{9}) x_{1}^{2}x_{2}^{2}x_{3} \\ & 
	+ K_1\k_6 \Big(  K_{2}^{2}K_{4}\k_{3}^{2}\k_{9}^{2}\, x_{1}^{4}x_{3}
	+2K_{2}K_{3}K_{4}\k_{3}^{2}\k_{9}\k_{12} \, x_{1}^{3}x_{2}x_{3} 
		 +K_{1}K_{2}K_{3}\k_{3}\k_{6}\k_{12}( \k_{9}+\k_{12}) x_{1}^{2}x_{2}^{3}  \\ &
	+K_{1}K_{2}K_{3}K_{4}\k_{3}\k_{6}\k_{9}\k_{12}\, x_{1}^{2}x_{2}^{2}
	+K_{1}K_{3}^{2}\k_{6}\k_{12}^{2} (\k_{3}+\k_{6})\,  x_{1}x_{2}^{4} 
	+2\,K_{1}K_{2}K_{3}\k_{3}\k_{6}\k_{12}^{2}\, x_{1}x_{2}^{3}x_{3} \\ &+K_{1}K_{2}K_{3}^{2}\k_{3}\k_{6}\k_{12}^{2}\, x_{1}x_{2}^{3}
	+K_{1}K_{3}^{2}\k_{6}^{2}\k_{12}^{2}\, x_{2}^{4}x_{3}+K_{1}^{2}K_{3}^{2}\k_{6}^{2}\k_{12}^{2}\,  x_{2}^{4}\Big).
	\end{split}
	\end{align}}

\begin{proposition}[\cite{conradi-mincheva,FeliuPlos}]\label{prop:multi}
With $p_\eta$ as in \eqref{eq:mypolynomial}, it holds:
\begin{itemize}[leftmargin=1.5cm]
	\item[(Mono)]  If $p_\eta(x)$ is positive for all $x_1,x_2,x_3>0$, then any $\k\in \pi^{-1}(\eta)$ does not enable multistationarity, and there is exactly one positive steady state in each invariant linear subspace.
	\item[(Mult)] If $p_\eta(x)$  is negative for some $x_1,x_2,x_3>0$, then any $\k\in \pi^{-1}(\eta)$ enables multistationarity in the invariant linear subspace containing the  point
	{\small \begin{align*} 
		\varphi(x_1,x_2,x_3) & = \left(x_1,x_2,x_3, \frac{K_2\k_3 x_1x_3}{K_1\k_6 x_2}, \frac{K_2K_4\k_3\k_9x_1^{2} x_3}{K_1K_3\k_6\k_{12}x_2^{2}},
		\frac{x_1x_3}{K_1}, \frac{\k_3x_1x_3}{K_1\k_6},  \frac{K_2\k_3x_1^{2}x_3}{K_1K_3\k_6 x_2}, \frac{K_2 K_3\k_3 \k_9 x_1^{2} x_3}{K_1\k_6\k_{12}x_2}\right).
		\end{align*}}%
\end{itemize}
\end{proposition}

\medskip
Explicitly, the polynomial $p_\eta$ equals $\det( J_F(\varphi(x_1,x_2,x_3))$, where $F\colon \R^9 \rightarrow \R^9$ is the function with first three components being the left-hand side of the equations in \eqref{eq:cons_laws}, and last $6$ components being the right-hand side of $\tfrac{dx_4}{dt},\dots,\tfrac{dx_9}{dt}$ in \eqref{eq:ode}, and $J_F$ denotes the corresponding Jacobian. 
The Brouwer degree of $p_\eta$ at zero is $1$, and this is used to derive conditions (Mono) and (Mult) above (see \cite{FeliuPlos}). Proposition~\ref{prop:multi} is a specific instance of a general theorem to identify multistationarity for networks satisfying three conditions, namely dissipativity, absence of boundary steady states, and existence of an algebraic parametrization of the steady states \cite{FeliuPlos}. Therefore, the approaches we use in this paper will likely be applicable to other relevant networks. 

In view of Proposition~\ref{prop:multi}, 
in order to determine what  reaction rate constants $\k$ enable multistationarity, we need to study what signs $p_\eta$ attains over $\R^3_{>0}$, as a function of $\eta$. 
To this end, we
study the relation between the coefficients of $p_\eta$ and the signs the polynomial attains using the Newton polytope of $p_\eta$ and a SONC decomposition, reviewed in the next subsection.

\subsection{The Newton Polytope, circuit polynomials, and signs}\label{sec:newton}
Key results on the relation between the coefficients of a polynomial and the signs the polynomial attains, build on a geometric object, namely the Newton polytope. 
Consider a polynomial  $p(x)=p(x_1,\ldots, x_n)=\sum_
\alpha c_{\alpha} x_1^{\alpha_1}\cdots x_n^{\alpha_n}$ in $\R[x_1, \ldots, x_n]$, where $\alpha=(\alpha_1, \ldots, \alpha_n)\in \Z^n_{\geq 0}$. 
The \emph{exponent set} of $p$ is the set of points $\alpha$  in $\Z^n_{\geq 0}$ such that $c_\alpha\neq 0$. 
The \emph{Newton polytope} $\newton{p}\subseteq \R^n$ associated with $p$ is the convex hull of the exponent set. Given a face $F$ of $\newton{p}$, we define the restriction  of $p$  to the monomials supported on $F$ as
\[p_F(x):=\sum_{\alpha \in F} c_{\alpha} x_1^{\alpha_1}\cdots x_n^{\alpha_n}.\]

The first main property of the Newton polytope is that any nonzero sign attained by $p_F(x)$ also is attained by $p(x)$. The following proposition is folklore in real algebraic geometry;
Remark~\ref{rk:outer_normal} sketches the proof by explicitly constructing the relevant points.

\begin{proposition}\label{prop:newton}
Let $p \in \R[x_1, \ldots, x_n]$. 
Given a nonempty face $ F$ of $\newton{p}$, consider the restriction $p_F$ of $p$ to the monomials supported on $F$. 
For any $x\in \R^n_{>0}$ such that $p_F(x)\neq 0$, there exists $y\in \R^n_{>0}$ such that
\[ \sign(p(y))= \sign(p_F(x)).\]
In particular, if the coefficient of one of the monomials supported on a vertex of $\newton p$ is negative, then there exists $x\in \R^n_{> 0}$ such that $p(x)<0$.
\end{proposition}

\begin{remark}\label{rk:outer_normal}
In the context of Proposition~\ref{prop:newton}, we find explicit values of $y$ where the sign of $p(y)$ agrees with the sign of $p_F(x)$ as follows.  
For $p(x)=\sum_
\alpha c_{\alpha} x_1^{\alpha_1}\cdots x_n^{\alpha_n} \in \R[x_1, \ldots, x_n]$, consider a $d$-dimensional face $F$ of $\newton p$ and assume $\newton p$ has dimension $n$. 
The \emph{outer normal cone $\mathcal{N}_F$ at the face $F$} is the cone generated by the outer normal vectors of the supporting hyperplanes of all the facets of $\newton p$ containing $F$. 
Then for any vector $v=(v_1,\ldots, v_n)$  in the interior of $\mathcal{N}_F$  (relative to the affine subspace of dimension $n-d$ containing it),  the scalar product $v\cdot x$
 for $x\in \newton p$ is maximized when $x$ belongs to the face $F$, where the value is a constant $c$ \cite{ziegler:polytopes}.
Hence, given $x\in \R^n_{>0}$, we have 
\[ p(x_1t^{v_1}, \ldots, x_nt^{v_n}) =   \sum_{\alpha} c_\alpha\, x^\alpha  \, t^{v_1 \alpha_1+\dots+v_n\alpha_n} =     p_{F}(x)\, t^{c} +\text{lower order terms in }t. \]
Hence, the sign of $p(x_1t^{v_1}, \ldots, x_nt^{v_n})$ agrees with the sign of $p_F(x)$  for $t\in \R_{>0}$ large enough.
\end{remark}

\begin{example}
	\label{eg:NP1}
	Consider the polynomial $p(x,y)= y-4 x y^3+x^2y^4+8 x^3 y^4$. The Newton polytope $\newton p$ is a quadrilateral in the plane, see left panel in Figure~\ref{Figure:NewtonPolytopes}. 
As $(1,3)$ is a vertex, $p(x,y)$ attains negative values over $\R^2_{>0}$ by Proposition~\ref{prop:newton}. 
To find a point where $p$ is negative, consider the outer normal cone at $(1,3)$, which is generated by the outer normal vectors $v_1:= (-2,1)$ and $v_2:=(-1,1)$.  The vector $u=v_1+v_2=(-3,2)$ belongs to its interior. Evaluation of $p$ at $( t^{-3}, t^2)$ is $-4t^3  +2t^2 +8 t^{-1}$, which is negative for $t$  larger than $\approx 1.34$. 
\end{example}

In what follows, a  point $\alpha$ in the exponent set of a polynomial $p\in \R[x_1,\dots,x_n]$  is said to be positive (negative) if the coefficient of the  monomial  $x^\alpha$ is positive (negative).  
A useful consequence of Proposition~\ref{prop:newton} is the following result.

\begin{corollary}\label{cor:mult}
Let $p \in \R[x_1, \ldots, x_n]$. Assume $\newton p$ has dimension $n$ and that all negative points of the exponent set of $p$ belong to some proper face of $N(p)$ (of dimension smaller than $n$). Then  the following  equivalence of statements holds:
\begin{center}
 $p(x)\geq 0$ for all $x\in \R^n_{>0}$ \qquad if and only if \qquad  $p(x)> 0$ for all $x\in \R^n_{>0}$.
\end{center}
\end{corollary}
\begin{proof}
The reverse implication is clear. To prove the forward implication, decompose  $p(x)$ as 
\[ p(x)= \sum_{\substack{\alpha \textrm{ in the boundary of }\newton p \\ c_\alpha\neq 0}}
  c_{\alpha} x_1^{\alpha_1}\cdots x_n^{\alpha_n} +  \sum_{\substack{\alpha \textrm{ in the interior of }\newton p \\ c_\alpha\neq 0}}
  c_{\alpha} x_1^{\alpha_1}\cdots x_n^{\alpha_n}.   \] 
By assumption, the second summand has only positive coefficients and hence is positive over $\R^n_{>0}$. 
If $p(x)=0$ for some $x\in \R^n_{>0}$, then necessarily the first summand is negative at this point $x$, and it follows that the restriction of $p$ to some proper face attains negative values. By Proposition~\ref{prop:newton}, the same holds for $p$, contradicting that $p(x)\geq 0$ for all $x\in \R^n_{>0}$. 
\end{proof}

\begin{figure}[t]
	\scalebox{0.6}{\hspace*{-1cm}
\begin{minipage}{0.3\textwidth}
\begin{tikzpicture}
\coordinate (Origin)   at (0,0);
\coordinate (XAxisMin) at (-2,0);
\coordinate (XAxisMax) at (4,0);
\coordinate (YAxisMin) at (0,-2);
\coordinate (YAxisMax) at (0,4);
\clip (-2,-2) rectangle (4cm,4cm);
\pgftransformcm{0.75}{0}{0}{0.75}{\pgfpoint{0cm}{0cm}}  
\draw [ultra thick,gray,-latex] (XAxisMin) -- (XAxisMax);
\draw [ultra thick,gray,-latex] (XAxisMax) -- (XAxisMin);
\draw [ultra thick,gray,-latex] (YAxisMin) -- (YAxisMax);
\draw [ultra thick,gray,-latex] (YAxisMax) -- (YAxisMin);
\draw[style=help lines,dashed] (-14,-14) grid[step=1cm] (14,14);
\node[inner sep=1pt, label={[label distance=-0.05cm]0:\textcolor{black}{\textbf{y}}}] at (0,5) {};
\node[inner sep=1pt, label={[label distance=-0.05cm]90:\textcolor{black}{\textbf{x}}}] at (5.1,0) {};
\node[draw,circle,blue,inner sep=2.5pt,fill] at (0,1) {};
\node[draw,circle,red,inner sep=2.5pt,fill] at (1,3) {};
\node[draw,circle,blue,inner sep=2.5pt,fill] at (2,4) {};
\node[draw,circle,blue,inner sep=2.5pt,fill] at (3,4) {};
\draw[line width=1pt] (0,1) -- (1,3) -- (2,4) -- (3,4) -- cycle;
\draw[line width=1.5pt,fill=blue,very nearly transparent] (0,1) -- (1,3) -- (2,4) -- (3,4) -- cycle;
\draw[line width=1.5pt,-latex,dashed,red] (1,3)--(-2,5);
\draw[line width=1.5pt,fill=red,very nearly transparent] (1,3) -- (-2,6) -- (-5,6) -- cycle;
\node[inner sep=1pt, label={[label distance=-0.1cm]45:\textcolor{red}{\small$\mathbf{u=(-3,2)}$}}] at (-2.7,3) {};
\end{tikzpicture}	
\end{minipage}
\hspace*{3cm}
\begin{minipage}{0.3\textwidth}
	\begin{tikzpicture}
	\coordinate (Origin)   at (0,0);
	\coordinate (XAxisMin) at (-2,0);
	\coordinate (XAxisMax) at (4,0);
	\coordinate (YAxisMin) at (0,-2);
	\coordinate (YAxisMax) at (0,4);
	\clip (-2,-2) rectangle (4cm,4cm);
	\pgftransformcm{0.75}{0}{0}{0.75}{\pgfpoint{0cm}{0cm}}  
	\draw [ultra thick,gray,-latex] (XAxisMin) -- (XAxisMax);
	\draw [ultra thick,gray,-latex] (XAxisMax) -- (XAxisMin);
	\draw [ultra thick,gray,-latex] (YAxisMin) -- (YAxisMax);
	\draw [ultra thick,gray,-latex] (YAxisMax) -- (YAxisMin);
	\draw[style=help lines,dashed] (-14,-14) grid[step=1cm] (14,14);
	\node[inner sep=1pt, label={[label distance=-0.05cm]0:\textcolor{black}{\textbf{y}}}] at (0,5) {};
	\node[inner sep=1pt, label={[label distance=-0.05cm]90:\textcolor{black}{\textbf{x}}}] at (5.1,0) {};
	\node[draw,circle,blue,inner sep=2.5pt,fill] at (2,4) {};
	\node[draw,circle,blue,inner sep=2.5pt,fill] at (4,2) {};
	\node[draw,circle,blue,inner sep=2.5pt,fill] at (0,0) {};
	\node[draw,circle,red,inner sep=2.5pt,fill] at (2,2) {};
	\draw[line width=1pt] (0,0) -- (2,4) -- (4,2) -- (0,0) -- cycle;
	\draw[line width=1.5pt,fill=blue,very nearly transparent] (0,0) -- (2,4) -- (4,2) -- (0,0) -- cycle;
	\end{tikzpicture}	
\end{minipage}}
	\caption[]{\small (Left) The quadrilateral corresponds is $\newton{p}$ for $p$ in Example~\ref{eg:NP1}, the shaded region is the outer normal cone at the vertex, and dashed vector is the chosen $u$. 
(Right) The triangle is the Newton polytope of the Motzkin polynomial  in Example~\ref{Example:Motzkin}.}	\label{Figure:NewtonPolytopes}
\end{figure}
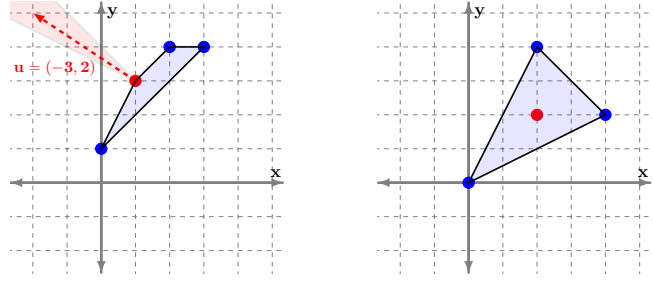

We review next \textbf{circuit polynomials}, an important tool to  derive conditions that guarantee a polynomial is nonnegative, that is, it does not attain negative values. Iliman and de Wolff introduced circuit polynomials in \cite{Iliman:deWolff:Circuits}, extending earlier work by Reznick \cite{Reznick:AGI}.

\begin{definition}
	\label{Definition:CircuitPolynomial}
	A polynomial $p\in \R[x_1,\dots,x_n]$ is  a \emph{circuit polynomial} if it is of the form
	\begin{align*}
	p({x}): = c_{{\beta}} {x}^{{\beta}} + \sum_{j=0}^{r} c_{{\alpha(j)}} {x}^{{\alpha(j)}}
	\end{align*} 
	with $r \leq n$, coefficients $c_{{\alpha(j)}} \in \R_{>0}$, $c_{{\beta}} \in \R$, and exponents ${\alpha(j)}, {\beta} \in \N^n$ such that $\newton{p}$ is a simplex with vertices ${\alpha(0)}, \dots, {\alpha(r)}$ containing ${\beta}$ in its interior.
	
	Every circuit polynomial $p$ has an associated \emph{circuit number}, $\Theta_p$, defined as
	\begin{align*}
	\Theta_p : = \prod_{j = 0}^{r} \left( \frac{c_{{\alpha(j)}}}{\lambda_j}\right)^{\lambda_j}
	\end{align*}
	where $\lambda_0, \dots ,\lambda_n$ are the unique barycentric coordinates of ${\beta}$ with respect to ${\alpha(0)}, \dots ,{\alpha(r)}$. That is, $\beta = \sum_{j=0}^r \lambda_j \alpha(j)$ with $0<\lambda_j \leq 1$ for $j=0,\dots,r$.
\end{definition} 

In contrast to the original definition of circuit polynomials given in \cite{Iliman:deWolff:Circuits}, we also allow $\alpha(j)$ to contain noneven entries in Definition~\ref{Definition:CircuitPolynomial}. 
The two definitions coincide when $x$ is restricted to the positive orthant, since one can consider $q(x_1,\dots,x_n)= p(x_1^2,\dots,x_n^2)$; for further details see e.g., the discussion in \cite[Section 3.1]{Iliman:deWolff:Circuits}.
With these considerations, the  theorem that follows is a straightforward consequence of \cite[Theorem 3.8]{Iliman:deWolff:Circuits}.   It gives a way to check the nonnegativity of a circuit polynomial $p$ over $\R^n_{>0}$ using the circuit number $\Theta_p$.

\begin{theorem}[\cite{Iliman:deWolff:Circuits}, Theorem 3.8]
	\label{Theorem:MainCircuitNonnegativity}
	A circuit polynomial $p$ given as in Definition~\ref{Definition:CircuitPolynomial} is nonnegative over $\R^n_{\geq 0}$ if and only if 
	\[ -c_{{\beta}} \leq  \Theta_p. \]
\end{theorem}

We conclude this subsection with an example to illustrate Theorem~\ref{Theorem:MainCircuitNonnegativity}. 

\begin{example}
	\label{Example:Motzkin}
	Consider the polynomial $p(x,y) = 1 +x^2y^4+x^4y^2-c\, x^2y^2$. Its Newton polytope is the triangle with the exponents $\{{\alpha(0)},{\alpha(1)},{\alpha(2)}\} = \{(0,0),(2,4),(4,2)\}$ as vertices, all of which have positive coefficients, see right panel of Figure~\ref{Figure:NewtonPolytopes}.
	The exponent ${\beta} = (2 , 2)$ is in the interior of $\newton{p}$, and its barycentric coordinates with respect to ${\alpha(0)}, {\alpha(1)}, {\alpha(2)}$ are $\tfrac{1}{3}, \tfrac{1}{3}, \tfrac{1}{3}$. We compute the circuit number:
	\begin{align*}
	\Theta_p = (3)^{\frac{1}{3}} \cdot (3)^{\frac{1}{3}} \cdot (3)^{\frac{1}{3}} = 3.
	\end{align*}
	Therefore, by Theorem~\ref{Theorem:MainCircuitNonnegativity},  $p$ is nonnegative over $\R^2_{\geq0}$ if and only if $c\leq 3$.
\end{example}

For $c = 3$ in Example~\ref{Example:Motzkin}, $p(x,y)$ is known as the Motzkin polynomial, which is a prominent example of nonnegative circuit polynomials. It is the first published example of a nonnegative polynomial that cannot be represented as a sum of squares of polynomials \cite{Motzkin:AMGMIneq}. 
For further details on nonnegative circuit polynomials see \cite{Iliman:deWolff:Circuits}, and e.g., \cite{Dressler:Iliman:deWolff:Positivstellensatz,Kurpisz:deWolff:NewDependenciesPolyOpt}.
See also \cite{pantea-jac}, where conditions for the positivity of multivariate polynomials were derived.

\begin{remark}\label{rk:homogeneous}
	In what follows we will  repeatedly  encounter homogeneous polynomials. Recall that a polynomial $p \in \R[x_1,\dots,x_n]$ is homogeneous if the total degree of all monomials is the same, say $d$.
	In this case, $p(\lambda x)=\lambda^d p(x)$ for any $\lambda\in \R$. Hence, the set of signs $p$ attains over $\R^n_{>0}$ agrees with the set of signs the polynomial $p(\lambda x)$ attains over $\R^n_{>0}$ for any choice of $\lambda>0$. In particular, we can set one of the variables to $1$, and study the signs of the resulting polynomial in the remaining $n-1$ variables. 
\end{remark}

\subsection{Back to our system}\label{sec:ineq}
We have now the ingredients to re-derive the conditions on the reaction rate constants that enable or preclude multistationarity given in \cite{conradi-mincheva} and to formulate the strategy to study the open cases. 
Recall the map $\pi$ from Subsection~\ref{sec:system} and that we write $\eta=(K_1,K_2,K_3,K_4,\k_3,\k_6,\k_9,\k_{12})$. Let
\begin{equation}\label{eq:ab}
 a(\eta)= \k_3\k_{12}-\k_6\k_9, \qquad b(\eta)=(K_{{2}}+K_{{3}})\k_{{3}}\k_{{12}}-(K_1+K_{{4}})\k_{{6}}\k_{{9}}.
\end{equation}
The coefficients of the polynomial $p_\eta$ given in \eqref{eq:mypolynomial}  in the variables $x=(x_1,x_2,x_3)$ are polynomials in the eight parameters $ K_1,K_2$, $K_3$, $K_4$, $\k_3$, $\k_6$, $\k_9,\k_{12}$. 
Five of these coefficients are positive multiples of $a(\eta)$, one is a positive multiple of $b(\eta)$, and the rest of the coefficients are positive. 

Of relevance is the monomial whose coefficient is multiple of $b(\eta)$, namely $x_1^2x_2^2x_3$, with exponent vector
\[m :=(2, 2, 1). \]
The Newton polytope of $p_\eta$ depends on whether $a(\eta)$ vanishes or not. 
If $a(\eta)\neq 0$, then  
$\newton{p_\eta}$  is depicted in the left and middle panels of Figure~\ref{fig:newton} and has $10$ vertices:
{\small \begin{align*}
\vertices{\newton{p_\eta}} = \ \big\{ &
(4, 0, 2 ),
(2, 2, 2 ), 
(4, 0, 1 ),
(3, 2, 1 ),
(2, 3, 1 ), 
(0, 4, 1 ), 
(2, 3, 0 ),
(2, 2, 0 ),
(1, 4, 0 ),
(0, 4, 0 )
\big\}.
\end{align*}}%
The  point $m=(2,2,1)$ is in the relative interior of the hexagonal face of $\newton{p_\eta}$ depicted in the middle panel of Figure~\ref{fig:newton}. The monomials with coefficient multiple of $a(\eta)$ are supported on the boundary of $\newton{p_\eta}$.

For $a(\eta)=0$, the corresponding Newton polytope is shown on the right panel of Figure~\ref{fig:newton}. 
Now $m$ is an interior point of an edge of $\newton{p_\eta}$. All other monomials have positive coefficient. 
The vertices of this Newton polytope are 
$
(4, 0, 1 ),
(2, 3, 0 ),
(2, 2, 0 ),
(1, 4, 0 ),
(0, 4, 1 ),
(0, 4, 0 ).
$

Let $H$ be the face of $\newton{p_\eta}$ containing $ {m}$: $H$ is a hexagonal 2-dimensional face of $\newton{p_\eta}$ if $a(\eta)\neq 0$, and a $1$-dimensional face  if $a(\eta)= 0$. Let $p_{\eta,H}$ be the polynomial supported on the face $H.$

\begin{proposition}\label{prop:summary}
Let $p_\eta$ be as in \eqref{eq:mypolynomial} and $a(\eta),b(\eta)$ as in \eqref{eq:ab}. 
\begin{enumerate}[label=(\roman*)]
\item $p_\eta(x)$ is either positive for all $x\in \R^3_{>0}$ or attains negative values over  $\R^3_{>0}$. Hence, $\k$ enables multistationarity if and only if $p_\eta$ attains negative values in $\R^3_{>0}$, where $\eta=\pi(\k)$.
\item Assume $a(\eta)\geq 0$. Then  $\k$ enables multistationarity if and only if $p_{\pi(\k),H}$ attains negative values over $\R^3_{>0}$. 
\item If $a(\eta)\geq 0$ and $b(\eta) \geq 0$, then any $\k\in \pi^{-1}(\eta)$ precludes multistationarity and there is one positive steady state in each invariant linear subspace defined by the equations \eqref{eq:cons_laws}. 
\item If $a(\eta)<0$, then any $\k\in \pi^{-1}(\eta)$ enables multistationarity. 
\end{enumerate}
\end{proposition}
\begin{proof}
(i) Follows from  Corollary~\ref{cor:mult} as coefficients of monomials supported on the interior of $\newton{p_\eta}$ are positive; (ii) Follows from (i) and Proposition~\ref{prop:newton}, as only $m\in H$ can be a negative point; 
(iii) As $p_\eta$ has only positive coefficients, the statement follows from (Mono) in  Proposition~\ref{prop:multi}; 
(iv)  In this case four of the vertices are negative. From Proposition~\ref{prop:newton} we conclude that    (Mult)  in Proposition~\ref{prop:multi} holds.
\end{proof}

\begin{figure}[t!]
\begin{center}
\begin{minipage}[h]{0.3\textwidth}
\includegraphics[scale=0.3]{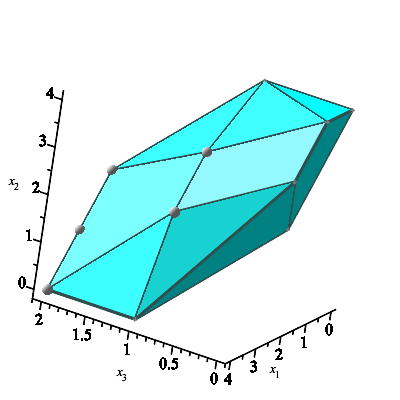}
\end{minipage}\quad
\begin{minipage}[h]{0.3\textwidth}
\includegraphics[scale=0.3]{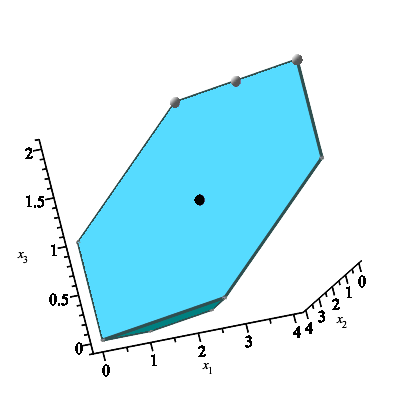}
\end{minipage}
\begin{minipage}[h]{0.3\textwidth}
\includegraphics[scale=0.3]{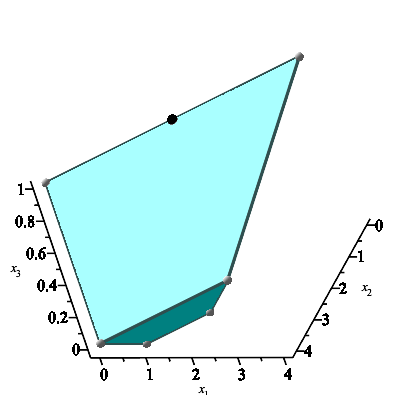}
\end{minipage}
\end{center}
\caption{{\small (Left and Middle) Newton polytope of the polynomial $p_\eta$ in \eqref{eq:mypolynomial} for $a(\eta)\neq 0$. The gray circles correspond to the monomials whose coefficient is a multiple of $a(\eta)$, and the black point to the monomial with coefficient a multiple of $b(\eta)$.  (Right) Newton polytope of $p_\eta$  when $a(\eta)=0$. The black point has coefficient a multiple of $b(\eta)$. }} \label{fig:newton}
\end{figure}

Statements (iii) and (iv) in Proposition~\ref{prop:summary} cover   the two known cases from \cite{conradi-mincheva}. 
As $m$ is not a vertex, $b(\eta)<0$ does not immediately guarantee that multistationarity is enabled. 

In view of Proposition~\ref{prop:summary}(i), whether $\k$ enables multistationarity or not only depends on $\pi(\k)$. Hence, we say that $\eta\in \R^8_{>0}$ enables multistationarity if this is the case for any $\k\in \pi^{-1}(\eta)$, or equivalently, if $p_\eta(x)$ attains negative values over $\R^3_{>0}$. 

\begin{corollary}
The set $X\subseteq \R^8_{>0}$ of parameter points $\eta$ that enable multistationarity is open with the Euclidian topology in $\R^8_{>0}$.
\end{corollary}
\begin{proof}
By Proposition~\ref{prop:summary}(i),  $\eta\in X$ if and only if $p_\eta(x^*)<0$ for some $x^*\in \R^3_{>0}$.  As $p_\eta$ is continuous in the coefficients, there exists an open ball centered at $\eta $ for which $p_{\eta'}(x^*)<0$ for any $\eta'$ in the ball.  Hence $X$ is open. 
\end{proof}
 
\begin{example}\label{ex:lifting}
Consider $\eta=(K_1,K_2,K_3,K_4,\k_3,\k_6,\k_9,\k_{12})= (343,1,1,1,2,1,1,1)$, for which $a(\eta)>0$ and
$p_{{\eta},H}(7,1,49)=-24706290<0$. By Proposition~\ref{prop:summary}(ii), $\eta$  enables multistationarity.
As   $\k=(1, 341, 2, 2,1,1,2,1,1,2,1,1)$ belongs to  $\pi^{-1}(\eta)$, it enables multistationarity.
 
In order to find a linear invariant subspace with multiple steady states, we use Remark~\ref{rk:outer_normal} 
to find a point where $p_\eta(x)<0$. 
To this end, we note that $(-1,-1,0)$ is an outer normal vector to $H$ and 
consider 
\[ p_\eta(7t^{-1},t^{-1},49)  = -\tfrac{24706290}{t^4} + \tfrac{38706521}{t^5}. \]
This expression is negative provided   $t>\tfrac{47}{30}$. With $t=2$, $p_\eta$ takes the value $-\tfrac{10706059}{32}<0$. 
Hence, the steady state defined by $(x_1,x_2,x_3)= (\tfrac{7}{2},\tfrac{1}{2},49)$ satisfies (Mult) in Proposition~\ref{prop:multi}.
This steady state is 
$x^*=\varphi(\tfrac{7}{2},\tfrac{1}{2},49)=(\tfrac{7}{2},\tfrac{1}{2},49,2,14,\tfrac{1}{2},1,7,7)$ and  belongs to the linear invariant    subspace defined by $E_{\rm tot } = 11$, $F_{\rm tot}=\tfrac{17}{2}$, $S_{\rm tot}= \tfrac{161}{2}$. We solve the equations for the positive steady states in this linear invariant subspace, and obtain $x^*$ together 
with   two other positive steady states, given approximately by: 
\begin{align*} 
(4.11, 0.91, 57.73,  1.51,  6.78, 0.7,  1.38,  6.2,  6.2), \ (3.43,  0.46, 47.21, 2.07, 15.6,  0.47, 0.94, 7.1,  7.1).
 \end{align*}
There are two other solutions with negative components.
We will see later in Example~\ref{ex:lifting2}, how the initial parameter value  and point $(7,1,49)$ were chosen. 

\end{example}

In what follows we study the open scenario  $a(\eta)\geq 0$ and $b(\eta)< 0$ by focusing on $p_{\eta,H}$, c.f. Proposition~\ref{prop:summary}(ii). 
We start by considering  two strategies to certify that $p_{\eta,H}(x)\geq 0$ for all $x\in \R^3_{>0}$, which imply that multistationarity is precluded. 
Afterwards, we show that  the polynomial $p_{\eta,H}(x)$ attains negative values for some $\eta$, and finally, we provide an explicit parametrization of the boundary between the region in the parameter space where multistationarity is enabled and the region where it is precluded. 
In particular, given any vector of parameters,  this gives a means to certify whether multistationarity is enabled.

\begin{remark}\label{rk:symmetry}
The ODE system in \eqref{eq:ode} is invariant under the map
	\begin{align*}
	(\k_1,\dots,\k_{12})  & \mapsto  (\k_{10},\k_{11},\k_{12},\k_7,\k_8,\k_9,\k_4,\k_5,\k_6,\k_1,\k_2,\k_3) \\
	(x_1,\dots,x_9)  & \mapsto  ( x_2,x_1,x_5,x_4,x_3,x_9,x_8,x_7,x_6).
	\end{align*}
	The reason is that the reaction network~\eqref{eq:network} remains invariant after interchanging $E$ with $F$, $S_0$ with $S_2$, the intermediate complexes accordingly, and relabeling the reactions as the map above indicates. Under this map, we have 
	\[ (K_1,K_2,K_3,K_4, \k_3,\k_{6},\k_9,\k_{12}) \xrightarrow{\ \sigma\ }  (K_4,K_3,K_2,K_1,\k_{12},\k_9,\k_6,\k_3). \]
	It follows that $\eta$ enables multistationarity if and only if $\sigma(\eta)$ does. 
	In particular, any relation on the parameters that guarantees or precludes multistationarity, 
	gives rise to a new relation after applying $\sigma$ to all parameters. In many cases though, the relations are already invariant by $\sigma$.
\end{remark}

\begin{remark}
Observe that $a(\eta)$ only depends on $\k_3,\k_6,\k_9,\k_{12}$. 
By letting 
\[ K=(K_1,K_2,K_3,K_4),\qquad \overline{\k} = (\k_3,\k_6,\k_9,\k_{12}),\]
it will be convenient sometimes to write 
$a(\overline{\k})$ instead of $a(\eta)$.
\end{remark}

\section{\bf The Case $a(\eta)\geq 0$ and $b(\eta)<0$: Monostationarity} \label{sec:apos}
We assume in this section that $a(\eta)\geq 0$ and $b(\eta)<0$ and recall the face $H$ of $\newton{p_\eta}$ defined in Subsection~\ref{sec:ineq}. By Proposition~\ref{prop:summary}(ii), $\eta$ enables multistationarity if and only if $p_{\eta,H}$ attains negative values over $\R^3_{>0}$. 
The face $H$ belongs to the hyperplane $x_1+x_2=4$, and hence $p_{\eta,H}$ is homogeneous of degree $4$ in $x_1,x_2$. 
Therefore, by Remark~\ref{rk:homogeneous}, it suffices to study the signs of $p_{\eta,H}$
after setting $x_2=1$. By abuse of notation, we denote the restricted polynomial by $p_{{\eta},H}(x_1,x_3)$. 
When $a(\eta)\neq 0$, we have
\begin{equation}\label{eq:pH}
\begin{aligned}
  p_{{\eta},H}(x_1,x_3) &  =  K_2\k_3 a(\eta) \Big(  K_{2}K_{4}  \k_{3}\k_{9} x_{1}^{4}x_{3}^{2}
  +K_{2}K_{3} \k_{3}\k_{12} x_{1}^{3}x_{3}^{2}+K_{1} K_{3}\k_{6}\k_{12} x_{1}^{2}x_{3}^{2}\Big) \\ & \hspace{-1cm}  +K_{1}K_{2}K_{3}\k_{3}\k_{6}\k_{12}\, b({\eta})\,  x_{1}^{2}x_{3}
    +K_{1} \k_6 \Big(  K_2 K_4\k_{3}\k_{9}  \big(  K_{2} \k_{3}\k_{9}\, x_{1}^{4}x_{3}   +2K_{3} \k_{3}\k_{12} \, x_{1}^{3}x_{3}\\ &
\hspace{-1cm} +K_{1}K_{3}\k_{6}\k_{12}\, x_{1}^{2} \big)
+ K_1K_3\k_{6}\k_{12}^{2} \big(K_{1}K_{3}\k_{6}  
 +2\, K_{2} \k_{3} \, x_{1}x_{3} + K_{2}K_{3}\k_{3}\, x_{1}
 + K_{3}\k_{6}\, x_{3} \big) \Big).
\end{aligned}
\end{equation}
When $a(\eta)=0$, the polynomial of interest is:
  \begin{align}\label{eq:pH2}
\begin{split}
p_{\eta,H}(x_1,x_3)& =K_{1} \k_6 \Big(   K_{2}K_{3}\k_{3}^2\k_{12}^2 ( (K_{2}+K_3)-(K_{1}+K_4)) x_{1}^{2} x_{3}  \\  & \hspace{-1cm} +
K_2K_4\k_{3}^{2}\k_{9} \big(  K_{2} \k_{9}\, x_{1}^{4}x_{3}    
 +2K_{3}\k_{12} \, x_{1}^{3} x_{3} \big) + K_1K_3 \k_{6}\k_{12}^{2} \big(  2\, K_{2}\k_{3}\, x_{1} x_{3} 
 +K_{3}\k_{6}\,  x_{3} \big) \Big).
\end{split}
\end{align}%
We derive two sufficient conditions for the nonnegativity of $p_{\eta,H}$: first, we consider the discriminant of a suitable polynomial (Subsection~\ref{sec:pol}), and then, circuit numbers (Subsection~\ref{sec:circuit}). The first strategy completely characterizes when $p_{\eta,H}$ is nonnegative when $a(\eta)=0$.

\subsection{Necessary polynomial condition for multistationarity via cylindrical algebraic decomposition.} \label{sec:pol}

The study of the discriminant of $p_{\eta,H}$ leads to the following theorem, whose proof relies on  symbolic algorithms from real algebraic geometry based on \cite{xiao:regularchain}. All computations are presented in the supplementary file {\it SupplInfo.mw}.

\begin{theorem}\label{prop:apos_disc}
Let $\eta\in \R^8_{>0}$ such that $a(\eta)\geq 0$ and $b(\eta)<0$. 
\begin{enumerate}[label=(\roman*)]
\item Consider the following polynomial:
\begin{footnotesize}
\begin{align*}
f(\eta):=& K_{2}^{2}K_{3}^{2}b(\eta)^{4}-K_{2}K_{3}\k_{3}\k_{12}( K_{1}K_{2}^{2}+K_{3}^{2}K_{4} ) b(\eta)^{3}+K_{1}K_{2}^{2}K_{3}^{2}K_{4} ( \k_{3}^{2}\k_{12}^{2}-20\k_{3}\k_{6}\k_{9}\k_{12}-8\k_{6}^{2}\k_{9}^{2}) b(\eta)^{2} \\
& + 18K_{1}K_{2}K_{3}K_{4}\k_{3}\k_{6}\k_{9}\k_{12} ( \k_{3}\k_{12}+2\k_{6}\k_{9} )  ( K_{1}K_{2}^{2}+K_{3}^{2}K_{4} ) b(\eta)   \\ &- K_{1}K_{4}\k_{6}\k_{9} \Big( 27\k_{3}^{2}\k_{6}\k_{9}\k_{12}^{2} (K_{1}^{2}K_{2}^{4}+K_{3}^{4}K_{4}^{2} )
+16K_{1}K_{2}^{2}K_{3}^{2}K_{4}   ( \k_{3}^{3}\k_{12}^{3}- \k_{6}^{3}\k_{9}^{3} )  \\ & +6K_{1}K_{2}^{2}K_{3}^{2}K_{4}\k_{3}\k_{6}\k_{9}\k_{12} ( \k_{3}\k_{12}+8\k_{6}\k_{9} )  \Big).
\end{align*}
\end{footnotesize}%
If  $f(\eta)\leq 0$, then $p_{\eta,H}$ is nonnegative over $\R^2_{>0}$, and  $\eta$  does not enable multistationarity.
\item Assume additionally that $a(\eta)=0$ and consider
 \begin{align*}
	g(K) & :=  K_2K_3 ( K_1+K_4 - K_2-K_3)^3 - 27 K_1K_4(K_2+K_3) (K_1  K_2- K_2  K_3+ K_3  K_4).
	\end{align*}
Then $p_{\eta,H}(x)$ is nonnegative over $\R^2_{>0}$  (and hence multistationarity is precluded) if and only if  $g(K)\leq 0$. 
Furthermore, $a(\eta)=0$ and $b(\eta)<0$ imply $K_1  K_2- K_2  K_3+ K_3  K_4>0$.
\end{enumerate}
\end{theorem}
\begin{proof}
We observe that the coefficient of $x_3$ in $p_{\eta,H}$ in \eqref{eq:pH} and \eqref{eq:pH2} 
is exactly $\k_6 K_1 q_{\eta}(x_1)$ with
\begin{align*}
q_{\eta}(x_1) & :=  K_2^2K_4\k_3^2\k_9^2\, x_1^4 + 2K_2K_3K_4\k_3^2\k_9\k_{12}\, x_1^3 \\& + K_2K_3\, b(\eta)\, \k_3\k_{12}\, x_1^2  +2K_1K_2K_3\k_3\k_6\k_{12}^2\, x_1 +K_1K_3^2\k_6^2\k_{12}^2.
\end{align*} 
When $a(\eta)=0$, $p_{\eta,H}(x_1,x_3)$ is exactly $\k_6 K_1 q_{\eta}(x_1)x_3$
and it follows that $q_\eta$ is nonnegative over $\R_{>0}$ \emph{if and only if} 
$p_{\eta,H}$ is nonnegative over $\R^2_{>0}$.
When $a(\eta)>0$, 
$p_{\eta,H}$ in \eqref{eq:pH} is a quadratic polynomial in $x_3$ with positive leading and constant terms.
Therefore, if $q_\eta$ is nonnegative over $\R_{>0}$, then $p_{\eta,H}$ is nonnegative over $\R^2_{>0}$.

Consequently, the theorem is proven if we show that:
(1) Assuming $a(\eta)\geq 0,b(\eta)<0$, $q_\eta$ is nonnegative over $\R_{>0}$ if and only if $f(\eta)\leq 0$, and (2) that this condition is equivalent to $g(K)\leq 0$ when additionally $a(\eta)=0$.

\smallskip
We prove (1). The polynomial $q_\eta$ has degree $4$ in $x_1$, and only the coefficient of $x_1^2$ is negative  (under the assumption $b(\eta)<0$).
By Descartes' rule of signs, $q_{\eta}$ has either two or zero positive  roots and either two or zero negative roots (counted with multiplicity). Therefore, $q_{\eta}$ attains negative values in $\R_{>0}$ if and only if $q_{\eta}$ has two distinct positive roots. 
 
Let $\Delta_{x_1}$ be the discriminant of $q_{\eta}$; it is a polynomial in $\eta$ and vanishes whenever $q_{\eta}$ has a multiple root. We restrict the parameter space to  the points where $b(\eta)<0$ and $a(\eta)\geq 0$ and define:
\begin{align*}
\Omega \ := \ \{\eta \in \R^8_{>0} \ : \ b(\eta)<0, a(\eta)\geq 0 \text{ and } \Delta_{x_1}(\eta) \neq 0\}.
\end{align*}
In each connected component  of $\Omega$,  the number of \emph{real roots} of $q_{\eta}$ is constant, and these are all simple roots. Since complex roots occur in pairs, the discriminant partitions $\R^8_{>0}$ into regions with four, two, or zero real roots. 
Now note that if $q_{\eta}$  has four real roots,  then necessarily two are positive and two are negative. Furthermore, in any connected component of $\Omega$ where $q_{\eta}$ has two real roots, these are either both positive or both negative for all $\eta\in \Omega$. This follows by continuity of the roots as a function of $\eta$ in each connnected component of $\Omega$, together with the fact that  $q_\eta$ cannot have a positive and a negative root with multiplicity $1$.
We conclude that in every connected component of $\Omega$, the number of \emph{positive real roots} of $q_{\eta}$ is also constant, and our goal is to  determine the components where this number is $2$. 

We compute    $\Delta_{x_1}$  and find that its zero set in $\Omega$ agrees with the zero set of one factor, $f$ in the statement. Hence the sign of $f(\eta)$ in each connected component of $\Omega$ is constant. 
So the strategy to prove (1) is to show that $q_\eta$ has two positive real roots if and only if $f(\eta)> 0$, by checking that this is the case for at least one point in each connected component of $\Omega$. 

To select such points, we will use the command {\tt SamplePoints} of the package {\tt RegularChains} in  {\tt Maple}, which builds upon the algorithms developed in \cite{xiao:regularchain}.
To reduce the computational cost to effectively find the points, we make some simplifications. 
We note first that $b(\eta),a(\eta)$ and $f(\eta)$ can be seen as polynomials in $K_1,K_2,K_3,K_4$ and the products $\k_3\k_{12}$ and $\k_6\k_9$, such that $f$ is homogeneous of degree $8$ in $K_1,K_2,K_3,K_4$ and homogeneous of degree $4$ in $\k_3\k_{12}$ and $\k_6\k_9$; $a(\eta)$ and $b(\eta)$ are both homogeneous of degree $1$ in $\k_3\k_{12}$ and $\k_6\k_9$; and $b(\eta)$ is homogeneous of degree $1$ in $K_1,K_2,K_3,K_4$.  Hence, given $\eta=(K_1,K_2,K_3,K_4,\k_3,\k_6,\k_9,\k_{12})$ and any $\lambda_1,\lambda_2,\lambda_3,\lambda_4>0$, the point
\[ \eta'=\Big( \lambda_1 K_1  , \lambda_1 K_2 , \lambda_1 K_3 , \lambda_1 K_4, \tfrac{\lambda_2\lambda_3}{\lambda_4}\k_3, \lambda_2\k_6, \lambda_3\k_9,\lambda_4\k_{12}\Big) \]
satisfies 
$f(\eta')=\lambda_1^8  \lambda_2^4\lambda_3^4 f(\eta)$, $a(\eta')= \lambda_2\lambda_3 a(\eta)$ and $b(\eta')= \lambda_1 \lambda_2\lambda_3 b(\eta)$. 
In particular, the signs of these three polynomials evaluated at $\eta$ and $\eta'$ agree, and $\eta$ belongs to $\Omega$, if and only if $\eta'$ does, in which case both belong to the same connected component. 
As a consequence, it is enough to consider points of the form 
$ \big( K_1,K_2, 1 , K_4, \k_3, 1, 1, 1\big) \in \Omega.$
The condition $a(\eta)\geq 0$ becomes $\k_3\geq 1$, and hence it is advantageous to reparametrize these points 
as $\big( K_1,K_2, 1 , K_4, a+1, 1, 1, 1\big)$ with $a\geq 0$. 

We have reduced the problem to selecting one point in each connected component of 
\[ \Omega' \ := \ \{ \eta=\big( K_1,K_2, 1, K_4, a+1, 1, 1, 1\big)   \ : \  K_1>0,K_2>0, K_4>0,a\geq 0, b(\eta)<0,  f(\eta) \neq 0
\}. \]
To this end, we consider $f(\eta)$ for $\eta\in \Omega'$ as a polynomial $f'_{v}(a)$ of degree 4 in $a$ and coefficients  in
$\R[K_1,K_2,K_4]$, where $v=(K_1,K_2,K_4)$. We  compute
 the discriminant $\Delta_a$ of $f'_{v}$ with respect to $a$, which is a polynomial in $K_1,K_2,K_4$.
The roots of the polynomial $f$ with variable $a$ deform continuously   in each connected component $C\subseteq \R^3_{>0}$ in the complement of $\Delta_a=0$. Specifically, for a given point $v$ in $C$, suppose $f_{v}$ has $r$ real roots $\{a_1, \ldots, a_r\}$ for $r\leq 4$ such that $a_i \leq a_{i+1}$ for all $i.$ For another point $v'$ in $C$, $f_{v'}$ also has $r$ roots $\{a'_1, \ldots, a'_r\}$ such that $a'_i \leq a'_{i+1}$ for all $i.$ In $C$ there exists a continuous path from $v$ to $v'$ such that $a_{i}$ deforms continuously to $a'_i.$ Therefore, there exists a continuous path in $\Omega'$ that takes a point from $v \times (a_i,a_{i+1})$ to $v' \times (a'_i,a'_{i+1}).$ 
 
Hence, in order to select at least one parameter point for each connected component of $\Omega'$, we consider  first (at least) one choice of $K_1,K_2,K_4>0$ in  each connected component $C$ of the complement of $\Delta_a=0$ with  the command {\tt SamplePoints}. We obtain a total of $22$ points. For each of them, we find the   nonnegative roots of $f$ as a polynomial in $a$, and then extend   $K_1,K_2,K_4$  to several parameter points in $ \Omega'$ by selecting one value of $a$ in each of the intervals the nonnegative roots define. 
This results in a list of points containing at least one point per connected component of $\Omega'$, and hence of $\Omega$.
Finally, for every such point $\eta$, we find the number of positive roots of $q_{\eta}$ (symbolically using the command {\tt RealRootCounting}) and determine the sign of $f(\eta)$. We conclude that  $q_{\eta}$ has two distinct positive real roots if and only if $f(\eta)>0$. 
It follows that $q_{\eta}$ is nonnegative in $\R_{>0}$ if and only if $f(\eta)\leq 0$, and in this case $p_{\eta,H}$ is nonnegative as well. This completes the proof of (1).

\smallskip
To prove (2), assume $a(\eta)=0$. It follows that $\k_3\k_{12} = \k_6\k_9$ and the condition $b(\eta)<0$ 
becomes $K_2+K_3 < K_1+K_4$. In this case, 
\[ f(\eta) = \k_6^4 \k_9^4(K_2 + K_3)(K_1K_2 - K_2K_3 + K_3 K_4)g(K). \]
Observe that under the assumption $b(\eta)<0$, we have
\[ K_1  K_2 + K_3  K_4>(K_1+K_4) \cdot \min \{K_2,K_3\} > (K_2+K_3) \cdot \min\{ K_2,K_3\} > K_2K_3.\] 
Hence $f(\eta)>0$ for $\eta\in\Omega$ such that $a(\eta)=0$ if and only if $g(K)>0$. 
This concludes the proof.
\end{proof} 
 	
	\begin{figure}[b!]
		\begin{center}
			\begin{minipage}[h]{0.4\textwidth}
		\includegraphics[scale=0.2]{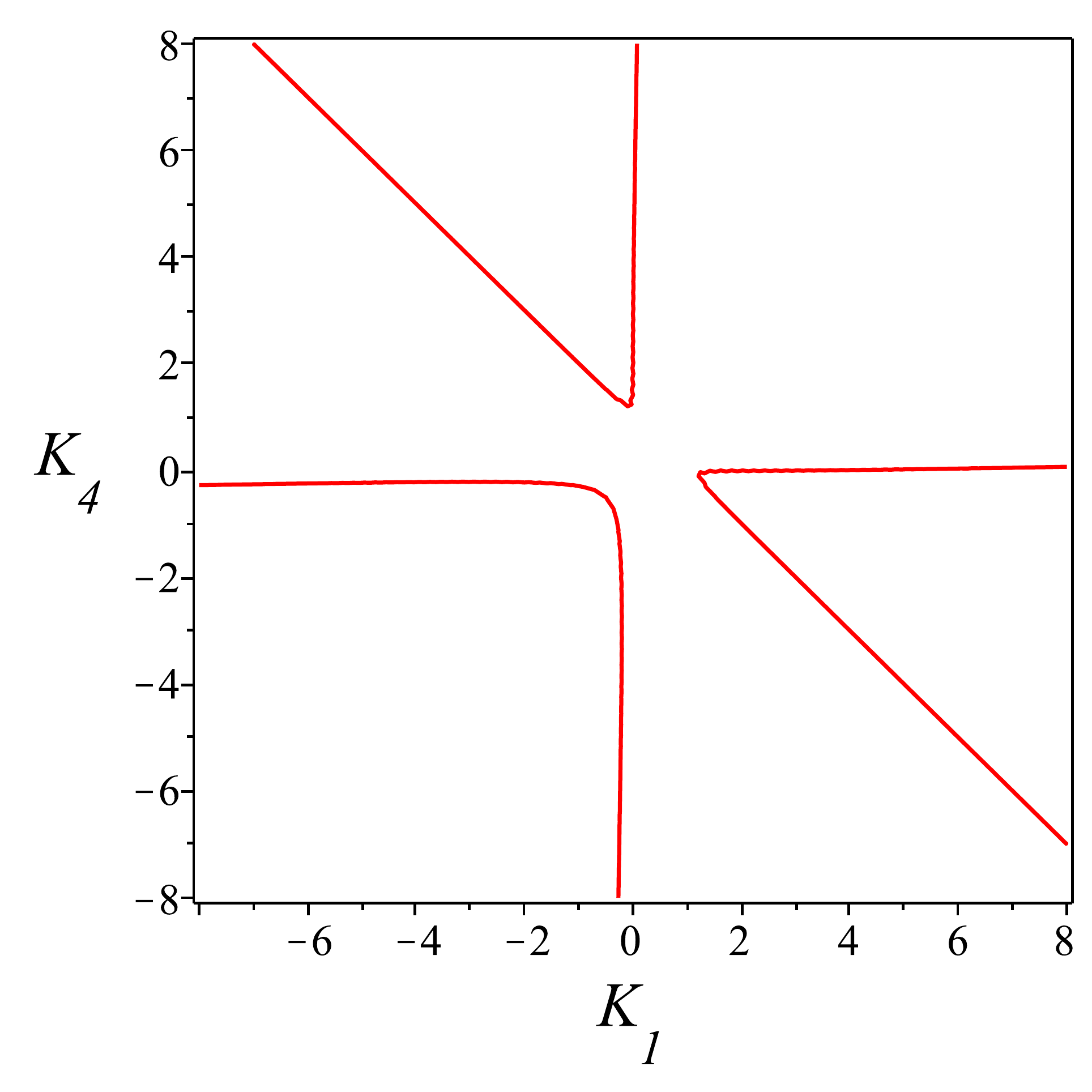}
			\end{minipage}
			\quad
			\begin{minipage}[h]{0.3\textwidth}
				\includegraphics[scale=0.18]{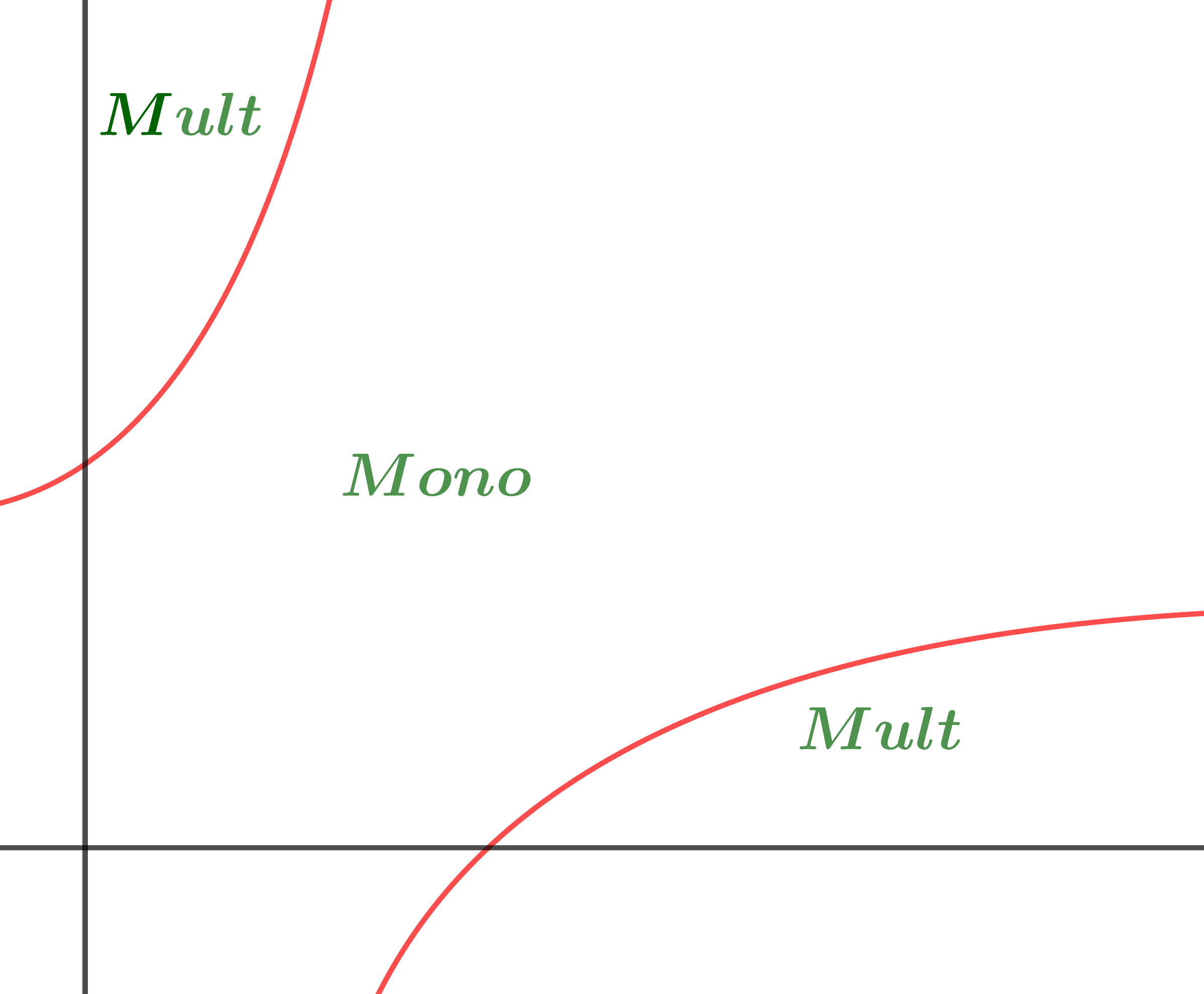}
			\end{minipage}
		\end{center}
		\caption{{\small (Left) 2-dimensional section of the zero set of the polynomial $g$ in Theorem~\ref{prop:apos_disc}. (Right) Cartoon picture of the partition of the positive orthant into the regions of mono- and multistationarity.}} \label{fig:cubic}
			\end{figure}

\begin{example}[Case $a(\eta)=0$] 
According to Theorem~\ref{prop:apos_disc}(ii), if $a(\eta)=0$, then multistationarity is characterized by the inequality $g>0$, which can be written as:
\[ K_2K_3 ( (K_1+K_4) - ( K_2+K_3))^3 > 27 K_1K_4(K_2+K_3) (K_1  K_2+ K_3  K_4- K_2  K_3).\]
The expressions at each side of the inequality are positive when $b(\eta)<0$.
We have $g(K_2+K_3,K_2,K_3,0)=g(0,K_2,K_3,K_2+K_3)=0$, meaning that $g=0$ intersects the two axes $K_1$ and $K_4$ at the given points. 

For example, let $K_2=K_3=1$. Then the zero set of the polynomial $g(K_1,1,1,K_4)$ in the $(K_1,K_4)$-plane is shown in Figure~\ref{fig:cubic}. The point $(K_2,K_3)=(1,1)$ gives a 2-dimensional slice of the zero set of the polynomial $g$ and its complement. By checking whether $g$ is positive or negative on points in the connected components of the complement of $g$, we find the cartoon depiction of the regions of multistationarity and monostationarity  illustrated in the right panel of Figure~\ref{fig:cubic}. 
	\end{example}

\begin{remark}
After setting $K_3=1$ as in the proof of Theorem~\ref{prop:apos_disc},  $g$  becomes a polynomial in $K_1,K_2$ and $K_4$. The degree of $g$ in $K_1$ and $K_4$ is $3$. The discriminant of $g$ with variables $K_1$ and $K_4$ is a polynomial in $K_2$, which does not vanish for any   $K_2>0$.  Therefore, for any $K_2>0$, the zero set of $g$ in the $(K_1,K_4)$-plane is as depicted in the left panel of Figure~\ref{fig:cubic}. 
\end{remark}

 \begin{example}\label{ex:apos1}
For any $\eta$ of the form $\eta=(K_1,1,1,K_4,2,1,1,1)$, we have $a(\eta)>0$ and 
\begin{small}
\begin{align*} 
f(\eta) &=3K_1^4 - 284K_1^3K_4 - 590K_1^2K_4^2 - 284K_1K_4^3 + 3K_4^4 - 40K_1^3 + 808K_1^2K_4 + 808K_1K_4^2\\ & - 40K_4^3 + 192K_1^2 - 320K_1K_4 + 192K_4^2 - 384K_1 - 384K_4 + 256. 
\end{align*}
\end{small}%
The solution set of $f=0$ in the $(K_1,K_4)$-plane is depicted in Figure~\ref{fig:apositive}, together with the monostationarity region given in Theorem~\ref{prop:apos_disc}. 
\end{example}

\begin{figure}[b!]
	\begin{center}
	
	\begin{minipage}[h]{0.4\textwidth}
	\includegraphics[scale=0.23]{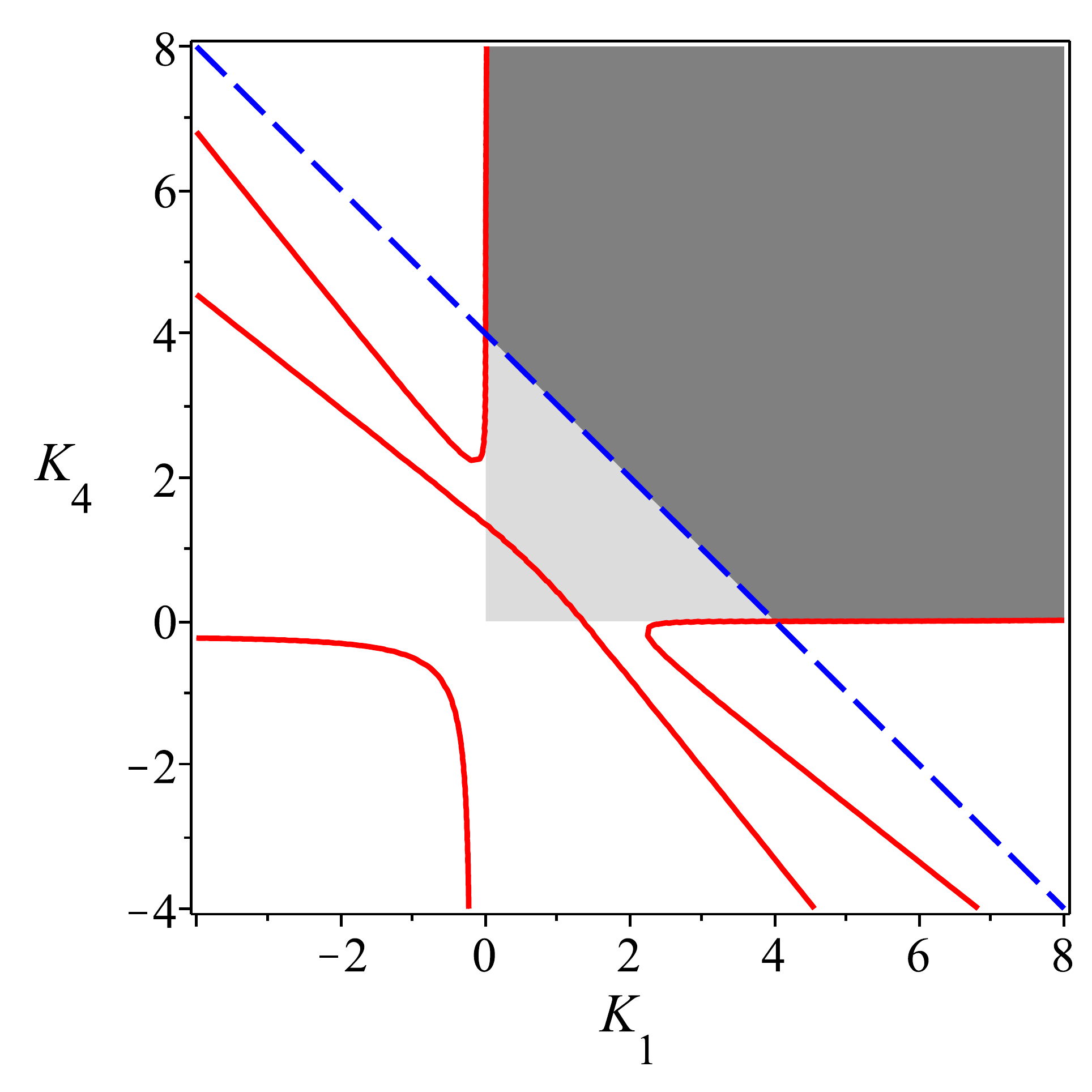}
\end{minipage}
\qquad
\begin{minipage}[h]{0.4\textwidth}
\includegraphics[scale=0.23]{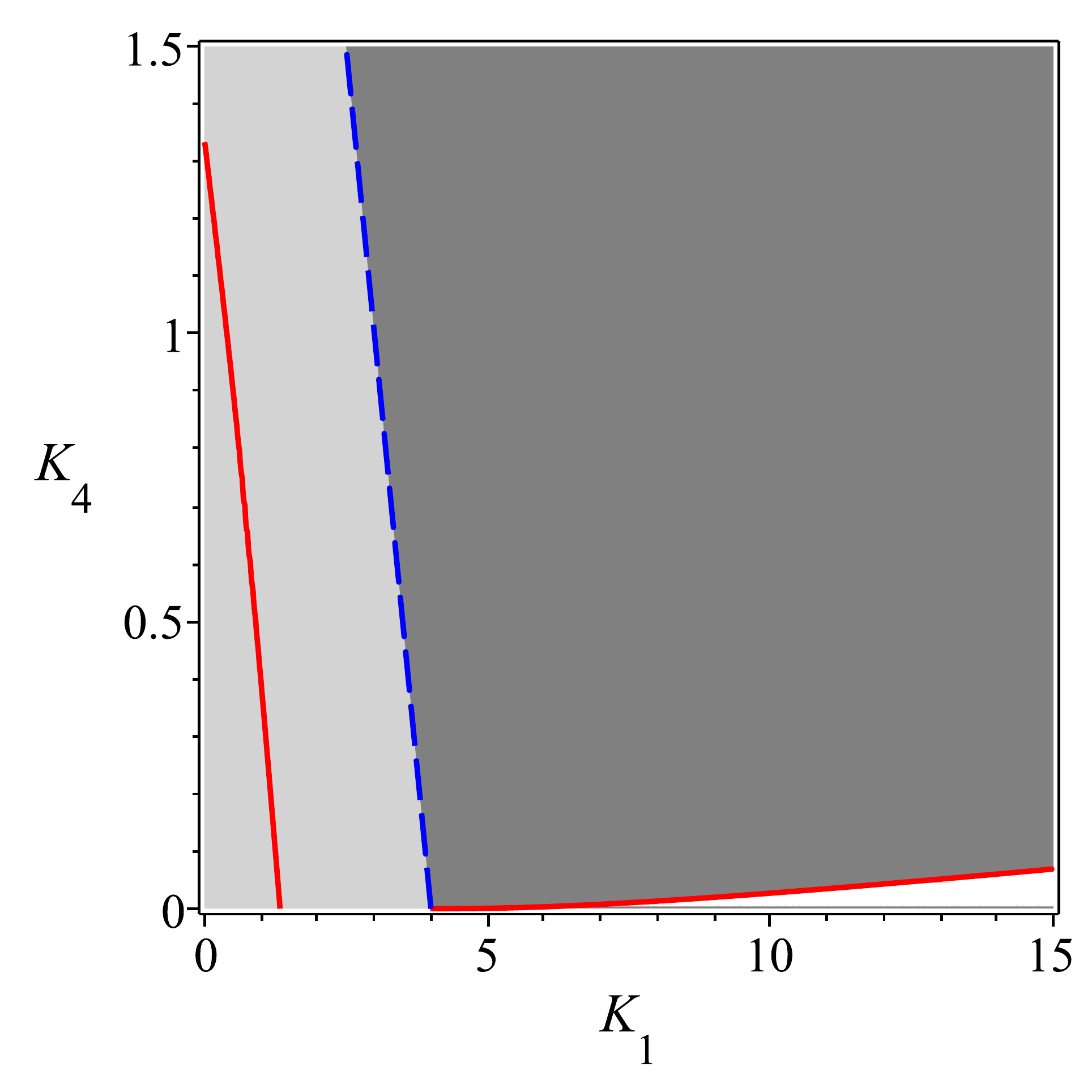}
\end{minipage}
		
	\end{center}
	
	\caption{{\small  $(K_2,K_3, \k_3,\k_6,\k_9,\k_{12})=(1,1,2,1,1,1)$, such that $a(\overline{\k})>0$. 
(Left) The solid-red curve is the solution set of $f=0$ in the $(K_1,K_4)$-plane, and the blue-dashed curve shows $b(\eta)=0$. In the gray region multistationarity is not enabled. The dark gray region is the one given in Theorem~\ref{prop:apos_disc}, where  $f<0$ and $b<0$. The light gray region shows $b\geq 0$ in the positive orthant. (Right) Zoom of the left panel for small $K_4$.
}} \label{fig:apositive}
	\end{figure}

\subsection{Necessary condition for multistationarity via circuit numbers. }\label{sec:circuit}
We now derive a necessary condition for multistationarity utilizing circuit polynomials. This new inequality, given in Theorem~\ref{prop:apos_circuit},  allows for an easier inspection of the points verifying it, compared to Theorem~\ref{prop:apos_disc} (c.f. Corollary~\ref{cor:apos_circuit}). 

As the case $a(\eta)=0$ is completely understood by Theorem~\ref{prop:apos_disc}, we focus mainly on the case $a(\eta)>0$ and $b(\eta)<0$. 
Consider the Newton polytope $H$ of $p_{{\eta},H}(x_1,x_3)$ in \eqref{eq:pH} for $a(\eta)\neq 0$. This polytope is the convex hull of the set $A_H$ of  exponent points, which we label as follows (see left panel of Figure~\ref{fig:newtonpoly}):
\begin{equation}\label{eq:barycenter}
\begin{aligned}
\alpha_1 & := (4, 2 ), & \alpha_2 &:=(2, 2 ), & \alpha_3&:= (0,  1 ), & \alpha_4 &:= (4, 1 ), & \alpha_5 &:=(2, 0 ),
& \alpha_6 &:= (0, 0 ), \\
m &:= (2, 1 ), & b_1 &:= (3,  2 ), & b_2 &:= (1,0), & i_1 &:= (3, 1 ), & i_2  &:= (1,  1 ). 
\end{aligned}
\end{equation}
 
 Note that $A_H$  is very well structured: ${m}$ is the barycenter of the two triangles given by the vertices ${\alpha_1}, {\alpha_3}, {\alpha_5}$ and ${\alpha_2}, {\alpha_4}, {\alpha_6}$;
${b_1}$ and ${b_2}$ are the midpoints of the two edges of $H$ given by ${\alpha_1}, {\alpha_2}$ and ${\alpha_5}, {\alpha_6}$ respectively; ${i_1}$ and ${i_2}$ are in the interior of $H$; and finally ${m}$ is the midpoint of both ${b_1}, {b_2}$ and ${i_1}, {i_2}$. 
We exploit this structure to decompose  $p_{{\eta},H}(x_1,x_3)$ into the sum of $4$ circuit polynomials with associated simplices with vertices $\{\alpha_1, {\alpha_3}, \alpha_5\}$,  $\{\alpha_2, \alpha_4, \alpha_6\}$, $\{b_1,b_2\}$ and $\{i_1,i_2\}$. 
 Afterwards we invoke Theorem~\ref{Theorem:MainCircuitNonnegativity} to derive conditions on the coefficients of   $p_{{\eta},H}(x_1,x_3)$  that guarantee the nonnegativity of this polynomial over $\R^2_{>0}$. 
This leads to the following theorem.

\begin{theorem} \label{prop:apos_circuit}
	Assume $a(\eta)\geq  0$ and $b(\eta)<0$. If
	\begin{align} 	\label{Inequality:SufficientInequality3}
	-b(\eta) \ \leq \ &   3\big(K_1K_4\k_{6}^2\k_{9}^2a(\eta)\big)^{\frac{1}{3}}  \Big(K_{1}^{\frac{1}{3}} +   K_{4}^{\frac{1}{3}}\Big) 
	   + 4 \left(K_{1}K_{4}\k_{3}\k_{6}\k_{9}\k_{12}\right)^{\frac{1}{2}}  +  2\left(K_{2}K_{3}\k_{3}\k_{12}a({\eta})\right)^{\frac{1}{2}},  
	\end{align}
	then $p_{\eta,H}$ is nonnegative over $\R^2_{>0}$, and hence $\eta$  does not enable multistationarity. 
\end{theorem}
\begin{proof}
Assume $a(\eta)>  0$.  
	We write $p_{{\eta},H}(x)$ as the sum of four circuit polynomials. Let $p_{{\eta},1}$ be a circuit polynomial which has the exponent ${m}$ as inner term and $2$-dimensional simplex ${\alpha_1}, {\alpha_3}, {\alpha_5}$ as follows,
	\begin{align*}
	p_{{\eta},1}(x_1,x_3) = c_{{\eta},{\alpha_1}} x^{{\alpha_1}} +  c_{{\eta},{\alpha_2}}x^{{\alpha_2}} +  c_{{\eta},{\alpha_3}} x^{{\alpha_3}} + \bar{c}_{{\eta},{1}}x^{{m}}
	\end{align*}
	where $c_{{\eta},{\alpha_i}}$ is exactly the coefficient of $x^{\alpha_i}$ in $p_{{\eta},H}(x)$, and $\bar{c}_{{\eta},{1}}$ is  in $\R$. Similarly, define the circuit polynomials $p_{{\eta},2}, p_{{\eta},3}, p_{{\eta},4}$ with exponent ${m}$ as inner term with $2$-dimensional simplex ${\alpha_2}, {\alpha_4}, {\alpha_6}$, and $1$-dimensional simplices ${b_1},{b_2}$ and ${i_1},{i_2}$ respectively. Let $\bar{c}_{{\eta},{i}}$ be the coefficient of $x^{{m}}$ in the respective polynomial $p_{{\eta},i}$. The Newton polytopes of these circuit polynomials are illustrated in the right panel of Figure~\ref{fig:circuits}.

	\begin{figure}[t!]
\begin{minipage}[h]{0.45\textwidth}
\scalebox{0.5}{\centering
\begin{tikzpicture}
\foreach \x in {0,1,2,...,4}
\foreach \y in {0,1,2} 
\draw[fill=none,draw=none] (5/2*\x,5/2*\y) circle (1.5pt) coordinate (m-\x-\y);

\draw[line width=1.5pt,fill=gray,ultra nearly transparent] (m-4-2) -- (m-2-2) -- (m-0-1) -- (m-0-0) -- (m-2-0) -- (m-4-1) -- cycle;
\draw  (m-4-2) -- (m-2-2) -- (m-0-1) -- (m-0-0) -- (m-2-0) -- (m-4-1) -- cycle;

\node[inner sep=3pt,circle,draw,fill,black,label={$\alpha_2$}] at (m-2-2) {};
\node[inner sep=3pt,circle,draw,fill,black,label={$\alpha_3$}] at (m-0-1) {};
\node[inner sep=3pt,circle,draw,fill,black,label={right:$\alpha_4$}] at (m-4-1) {};
\node[inner sep=3pt,circle,draw,fill,black,label={$\alpha_5$}] at (m-2-0) {};
\node[inner sep=3pt,circle,draw,fill,black,label={left:$\alpha_6$}] at (m-0-0) {};
\node[inner sep=3pt,circle,draw,fill,black,label={$\alpha_1$}, ] at (m-4-2) {};

\node[inner sep=3pt,circle,draw,fill,black,label={$b_1$}] at (m-3-2) {};
\node[inner sep=3pt,circle,draw,fill,black,label={$b_2$}] at (m-1-0) {};

\node[inner sep=3pt,circle,draw,fill,black,label={$i_2$}] at (m-1-1) {};
\node[inner sep=3pt,circle,draw,fill,black,label={$i_1$}] at (m-3-1) {};

\node[inner sep=3pt,circle,draw,fill,red,label={$m$}] at (m-2-1) {};
\end{tikzpicture}

%
}

\end{minipage}\qquad
\begin{minipage}[h]{0.45\textwidth}
		\scalebox{0.5}{\centering
\begin{tikzpicture}
\foreach \x in {0,1,2,...,4}
\foreach \y in {0,1,2} 
\draw[fill=none,draw=none] (5/2*\x,5/2*\y) circle (1.5pt) coordinate (m-\x-\y);

\draw[line width=1.5pt,fill=gray,ultra nearly transparent] (m-4-2) -- (m-2-2) -- (m-0-1) -- (m-0-0) -- (m-2-0) -- (m-4-1) -- cycle;
\draw  (m-4-2) -- (m-2-2) -- (m-0-1) -- (m-0-0) -- (m-2-0) -- (m-4-1) -- cycle;

\draw[thick,red,fill=red,ultra nearly transparent] (m-4-2) -- (m-0-1) -- (m-2-0) -- cycle;
\draw[thick,red] (m-4-2) -- (m-0-1) -- (m-2-0) -- cycle;
\draw[thick,orange,fill=orange,very nearly transparent] (m-2-2) -- (m-4-1) -- (m-0-0) -- cycle;
\draw[thick,orange] (m-2-2) -- (m-4-1) -- (m-0-0) -- cycle;
\draw[line width=2pt,blue,dashed] (m-3-2) -- (m-1-0);
\draw[line width=2pt,green!50!black,dashed] (m-3-1) -- (m-1-1);

\node[inner sep=3pt,circle,draw,fill,red,label={$m$}] at (m-2-1) {};
\node[inner sep=3pt,circle,draw,fill,black,label={$\alpha_2$}] at (m-2-2) {};
\node[inner sep=3pt,circle,draw,fill,black,label={$\alpha_3$}] at (m-0-1) {};
\node[inner sep=3pt,circle,draw,fill,black,label={right:$\alpha_4$}] at (m-4-1) {};
\node[inner sep=3pt,circle,draw,fill,black,label={$\alpha_5$}] at (m-2-0) {};
\node[inner sep=3pt,circle,draw,fill,black,label={left:$\alpha_6$}] at (m-0-0) {};
\node[inner sep=3pt,circle,draw,fill,black,label={$\alpha_1$}, ] at (m-4-2) {};

\node[inner sep=3pt,circle,draw,fill,black,label={$b_1$}] at (m-3-2) {};
\node[inner sep=3pt,circle,draw,fill,black,label={$b_2$}] at (m-1-0) {};

\node[inner sep=3pt,circle,draw,fill,black,label={$i_2$}] at (m-1-1) {};
\node[inner sep=3pt,circle,draw,fill,black,label={$i_1$}] at (m-3-1) {};
\end{tikzpicture}

%
}
\end{minipage}

		\caption{(Left) An illustration of $H$, where ${\alpha_j}, {b_j}, {i_j}$ are  as in \eqref{eq:barycenter}. Right panel: The circuits of the SONC decomposition, consisting of two 2-dimensional circuits with vertices $\alpha_1,\alpha_3,\alpha_5$ and $\alpha_2,\alpha_4,\alpha_6$ respectively, and two 1-dimensional circuits, with vertices $b_1,b_2$ and $i_1,i_2$ respectively.}
		\label{fig:circuits}\label{fig:newtonpoly}
	\end{figure}
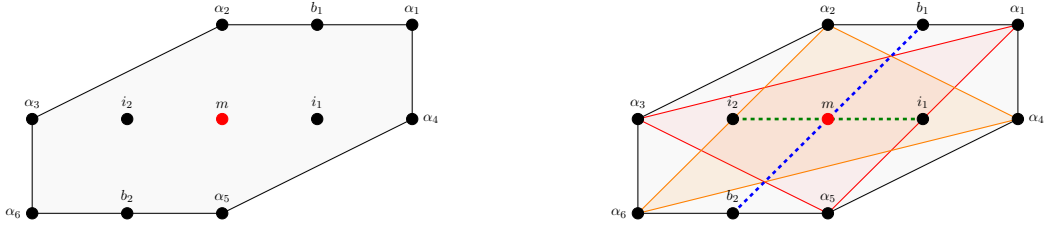
	
	The circuit number corresponding to each of the circuit polynomials are:
	\begin{align*}
	\Theta_{p_{{\eta},1}} &= 3(c_{{\eta},{\alpha_1}} c_{{\eta},{\alpha_3}} c_{{\eta},{\alpha_5}})^{\frac{1}{3}}, &
	\Theta_{p_{{\eta},2}} &= 3(c_{{\eta},{\alpha_2}} c_{{\eta},{\alpha_4}} c_{{\eta},{\alpha_6}})^{\frac{1}{3}}\\
	\Theta_{p_{{\eta},3}} &= 2(c_{{\eta},{b_1}} c_{{\eta},{b_2}} )^{\frac{1}{2}} &
	\Theta_{p_{{\eta},4}}  &= 2(c_{{\eta},{i_1}} c_{{\eta},{i_2}} )^{\frac{1}{2}}.
	\end{align*}
	Now assume that the following inequality is satisfied for $c_{{\eta},{m}}$, the coefficient of $x^m$ in $p_{\eta,H}$:
	\begin{align}
	\label{Inequality:SufficientInequality}
	-c_{{\eta},{m}} \leq \Theta_{p_{{\eta},1}} + \Theta_{p_{{\eta},2}} + \Theta_{p_{{\eta},3}} + \Theta_{p_{{\eta},4}}.
	\end{align}
	Then one can find $\bar{c}_{{\eta},{1}}, \bar{c}_{{\eta},{2}}, \bar{c}_{{\eta},{3}}, \bar{c}_{{\eta},{4}} \in \R$ such that $\sum \bar{c}_{{\eta},{i}} = c_{{\eta},{m}}$ and for all $i$, $-\bar{c}_{{\eta},{i}} \leq \Theta_{p_{\eta,i}}$. Theorem~\ref{Theorem:MainCircuitNonnegativity} implies that each $p_{{\eta},i}$ is nonnegative over $\R^2_{>0}$. As $p_{\eta,H}=
	p_{{\eta},1}+p_{{\eta},2}+p_{{\eta},3}+p_{{\eta},4}$, $p_{{\eta},H}$ also is nonnegative. In terms of the entries of $\eta$,  \eqref{Inequality:SufficientInequality} becomes
	{\small \begin{align*}
		-K_{1}K_{2}K_{3}\k_{3}\k_{6}\k_{12}b(\eta) \ \leq \ &  3K_{1}K_{2}K_{3}\k_{3}\k_6 \k_{12}\left(K_1 K_{4}^2\k_{6}^2\k_{9}^2a(\eta)\right)^{1/3}  \ + \ 3\left(K_{1}^5K_{2}^3K_{3}^3K_{4}\k_{3}^3\k_{6}^5\k_{9}^2\k_{12}^3a(\eta)\right)^{1/3} \\ & \ + \ 4\left(K_{1}^3K_{2}^2K_{3}^2K_{4}\k_{3}^3\k_{6}^3\k_{9}\k_{12}^3\right)^{1/2}    \ + \ 2\left(K_{1}^2K_{2}^3K_{3}^3\k_{3}^3\k_{6}^2\k_{12}^3a(\eta)\right)^{1/2}, \end{align*}}%
	which after factoring out terms and simplifying gives the inequality in the statement. 
	
	When $a(\eta)=0$, inequality~\eqref{Inequality:SufficientInequality3} reduces to $K_1+K_4-K_2-K_3 \leq 4\sqrt{K_1K_4}$. 
	We verify using  the function {\tt IsEmpty} in {\tt Maple 2019} that whenever this holds, then $g$ in Theorem~\ref{prop:apos_disc} is negative, implying that $p_{\eta,H}(x)$ is nonnegative over $\R^2_{>0}$.
	\end{proof}

\begin{remark} 
	The SONC decomposition of $p_{{\eta}}$ into $p_{\eta,1}, p_{\eta,2}, p_{\eta,3}, p_{\eta,4}$ in the proof of Theorem~\ref{prop:apos_circuit} is not unique. Other sufficient conditions may be derived using other covers of $H$, see e.g., \cite[page 20]{DresslerIlimandeWoldd:ConstrainedOptViaCircuitPoly}. Two main reasons underlie the choice of this particular cover. First, it uses the least possible number of circuits while using every  positive point only once. Hence we use all the possible positive weight and avoid introducing new parameters for nondisjoint circuits. Second, as ${m}$ is the barycenter of each chosen circuit, the derived circuit numbers have simple expressions.
\end{remark}

\begin{example} 
	To illustrate the use of inequality \eqref{Inequality:SufficientInequality3} to certify monostationarity, consider  $\eta =(2,0.5,0.5,2,2,1,1,1)$.
	Then,  \eqref{Inequality:SufficientInequality3} holds since the right hand side is  $\approx 24.72$, while the left hand side is 2.
By Theorem~\ref{prop:apos_circuit},  $\eta$ does not enable multistationarity.
	Indeed, $p_{\eta,H}({x_1,x_3}) \geq 0 $ for all ${x} \in \R^2_{\geq 0}$, as it also can be seen by rewriting the polynomial as:
	\begin{align*}
	p_{{\eta,H}}(x_1,x_3)   & =   x_{2}^4x_{3} + 4 x_{1}^4x_{3}   + \tfrac{1}{2}x_{1}^3x_{2}x_{3}^2 + 8x_{1}^3x_{2}x_{3}+ x_{1}^2x_{2}^2x_{3}^2\\
	& \qquad + 4x_{1}^2x_{2}^2  + 4x_{1}x_{2}^3x_{3} + x_{1}x_{2}^3 + x_{1}^4x_{3}^2 + x_{2}^4 + (x_{1}^2x_{3}-x_{2}^2)^2.
	\end{align*}
\end{example}

\begin{example}  
	We fix the parameters  $(K_2,K_3, \k_3,\k_6,\k_9,\k_{12})=(1,1,2,1,1,1)$ as in Example~\ref{ex:apos1}. 
	Figure~\ref{fig:regions} shows a  comparison of   the two necessary conditions for multistationarity from Theorem~\ref{prop:apos_circuit} and Theorem~\ref{prop:apos_disc}. 
	For this choice of parameters, inequality  \eqref{Inequality:SufficientInequality3} becomes
	\begin{align}
	\label{Inequality:K1K4plane}
	0 \leq 3 \left(K_1 K_4^2\right)^{\frac{1}{3}} \ + \ 4\sqrt{2} \left(K_1K_4\right)^{\frac{1}{2}} \ + \ 3 \left(K_1^2K_4\right)^{\frac{1}{3}} - K_1 - K_4  \ + \ 2\sqrt{2} \ + \ 4.
	\end{align}
Figure~\ref{fig:regions} hints at that the sufficient condition for monostationarity of Theorem~\ref{prop:apos_circuit} includes a cone pointed at zero. To investigate this further, consider the line $s K_1 =K_4$ for $s \in (0,+\infty)$. Then the right hand side of \eqref{Inequality:K1K4plane}  becomes
	\begin{align}
	\label{PolynomialK1}
	\big(3s^{\frac{2}{3}}+4\sqrt{2}s^{\frac{1}{2}}+3s^{\frac{1}{3}}-s-1\big) K_1 + \big(2\sqrt{2} + 4 \big).
	\end{align}
The positive semiline belongs to the monostationarity region if \eqref{PolynomialK1} is positive for all $K_1>0$. As  \eqref{PolynomialK1} is linear in $K_1$ with positive constant term, it is positive for all  $K_1>0$ if and only if the leading coefficient is positive. This holds if and only if $s$ lies in the interval $\approx (1/197.995,197.995).$
\end{example}

\begin{figure}[t!]
	\begin{center}
		\begin{minipage}[h]{0.4\textwidth}
			\includegraphics[scale=0.23]{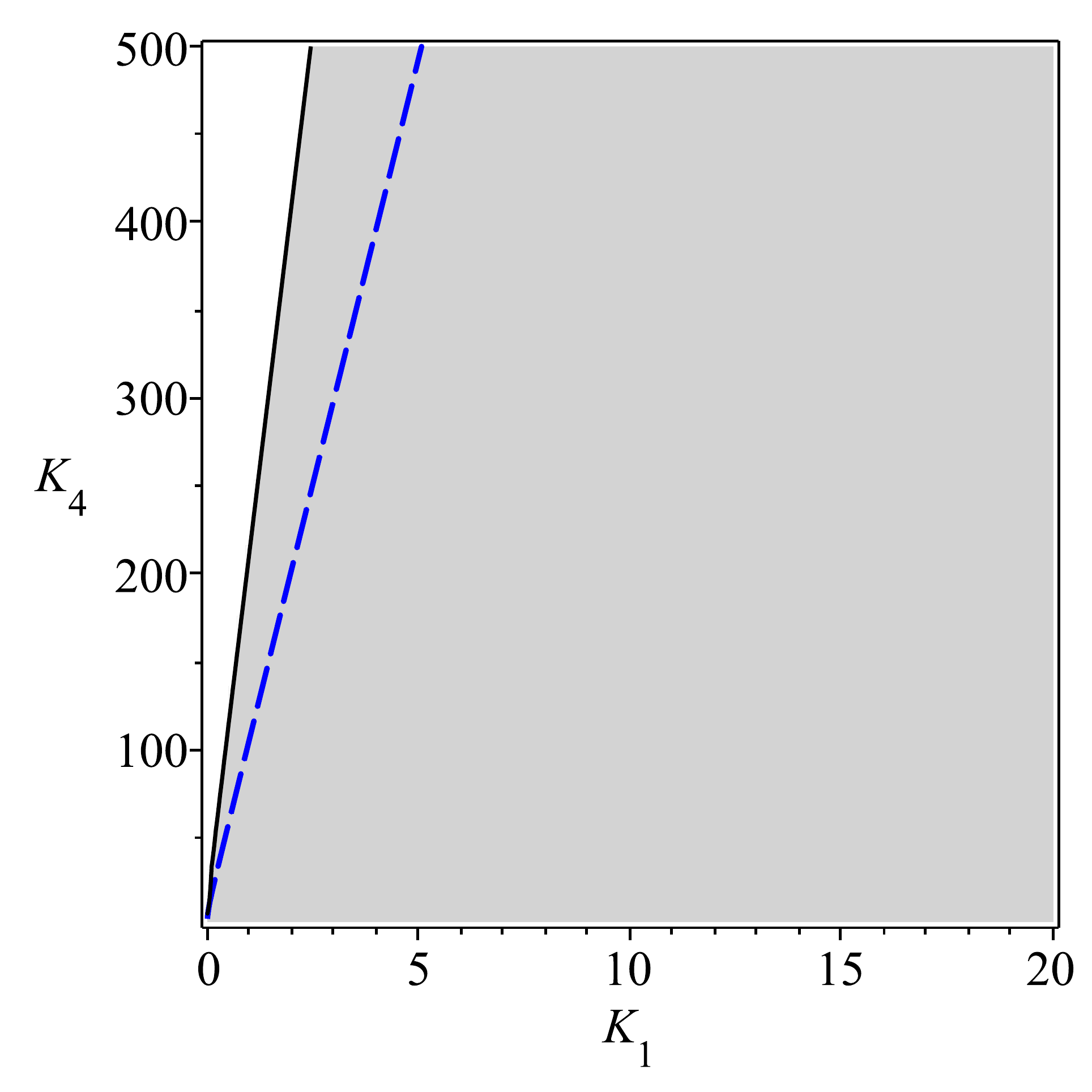}
		\end{minipage}
		\qquad
		\begin{minipage}[h]{0.4\textwidth}
			\includegraphics[scale=0.23]{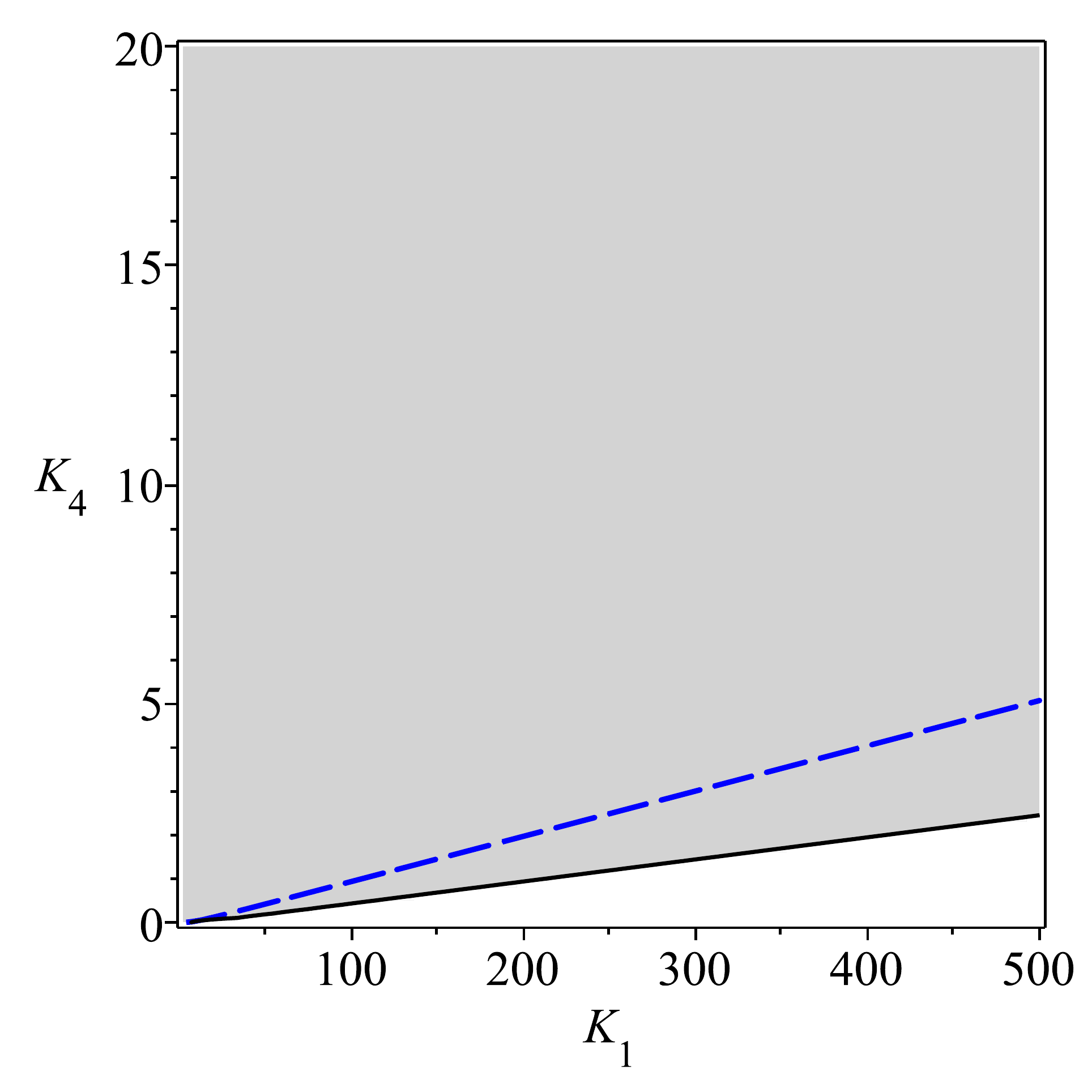}
		\end{minipage}
	\end{center}
	\caption{{\small For $(K_2,K_3, \k_3,\k_6,\k_9,\k_{12})=(1,1,2,1,1,1)$, the region between blue dashed lines is the region where we can certify monostationarity using Theorem~\ref{prop:apos_disc}. The region given between full lines is the region where we can certify monostationarity using Theorem~\ref{prop:apos_circuit}. The two panels focus on either $K_1$ large or $K_1$ small. }} \label{fig:regions}
\end{figure}

The conclusions in the example above extend to any choice of fixed parameters 
 $K_2,K_3, \k_3,\k_6$, $\k_9,\k_{12}$. In particular, in the $(K_1,K_4)$ plane, the region of monostationarity includes a cone pointed at zero that includes the line $K_1=K_4$. This is the content of the next corollary. This result will be critical to obtain a parametric description of the regions of mono- and multistationarity  in Section~\ref{sec:enabled}.

\begin{corollary}\label{cor:apos_circuit}
	Assume $\eta':=(K_2,K_3, \k_3,\k_6,\k_9,\k_{12})$ fixed such that $a(\overline{\k})\geq 0$ and consider the line $K_4=s K_1$ in $\R^2_{>0}$ with coordinates $K_1,K_4$.
	There exist $0<s_1(\eta')<\tfrac{14- \sqrt{192}}{2} $ and $\tfrac{14 +  \sqrt{192}}{2}<s_2(\eta')$ such that:
	\begin{itemize}
		\item[(i)]  For any $s\in [s_1(\eta'),s_2(\eta')]$, the points in the line $K_4=s K_1$ satisfy inequality 
 \eqref{Inequality:SufficientInequality3}. 
		\item[(ii)] If $s\notin [s_1(\eta'),s_2(\eta')]$, then there exists $K_1'$ such that \eqref{Inequality:SufficientInequality3} holds if and only if $K_1\leq K_1'$.
		\item[(iii)] If $\k_3\k_{12}$ increases, while $K_2,K_3,\k_6,\k_9$ remain fixed, then $s_1(\eta')$ decreases to zero and $s_2(\eta')$ increases to $+\infty$. 
	\end{itemize}
	In particular, if $K_1=K_4$ and $a(\overline{\k})\geq 0$, multistationarity is not enabled.
\end{corollary}
\begin{proof}
	As $\eta'$ is fixed, inequality \eqref{Inequality:SufficientInequality3} is a relation on   $K_1$ and $K_4$. We rewrite it as:
{\small \begin{align*}
		0 \ \leq \ &  - (K_1+K_4)\k_{{6}}\k_{{9}} + 3\big(K_1K_4\k_{6}^2\k_{9}^2a(\eta)\big)^{\frac{1}{3}}  \Big(K_{4}^{\frac{1}{3}} +   K_{1}^{\frac{1}{3}}\Big) 
	 \ + \  4 \left(K_{1}K_{4}\k_{3}\k_{6}\k_{9}\k_{12}\right)^{\frac{1}{2}} 	\\ &  +  2\left(K_{2}K_{3}\k_{3}\k_{12}a({\eta})\right)^{\frac{1}{2}} +(K_{{2}}+K_{{3}})\k_{{3}}\k_{{12}}.   \nonumber
		\end{align*}}%
	When $K_4=sK_1$, this inequality becomes
	{\small \begin{align} \label{eq:suff4}
		0  \ \leq \ & \Big( - (1+s)\k_{{6}}\k_{{9}} + 3 (s\k_{6}^2\k_{9}^2a({\eta}))^{\frac{1}{3}}  (s^{\frac{1}{3}}+1)\ + \ 4 (s \k_{3}\k_{6}\k_{9}\k_{12})^{\frac{1}{2}} \Big) K_1 \\ &   \ + \ 2(K_{2}K_{3}\k_{3}\k_{12}a({\eta}))^{\frac{1}{2}}+(K_{{2}}+K_{{3}})\k_{{3}}\k_{{12}}.   \nonumber
		\end{align} }%
	First, note that since by assumption $\k_3\k_{12}\geq \k_6\k_9$, we have:
	\[ (1+s)\k_{{6}}\k_{{9}} = (1+ s) \big( \k_{{6}}^2\k_{{9}}^2\big)^{\frac{1}{2}} \leq  (1+s) ( \k_{3}\k_{6}\k_{9}\k_{12} )^{\frac{1}{2}}.  \]
	Hence, if $(1+s) ( \k_{3}\k_{6}\k_{9}\k_{12} )^{\frac{1}{2}}\leq  4 (s \k_{3}\k_{6}\k_{9}\k_{12})^{\frac{1}{2}}$, then \eqref{eq:suff4} holds for all $K_1>0$. This inequality simplifies to $1+s \leq  4  \sqrt{s}$, which holds if and only if  
	$s\in \big( \tfrac{14- \sqrt{192}}{2},  \tfrac{14+ \sqrt{192}}{2}\big)$.

	Now,  inequality \eqref{eq:suff4} holds for all $K_1>0$ if and only if the coefficient of $K_1$ is nonnegative. 
	We set $r^6 = s$, and the coefficient of $K_1$ becomes 
	\[
	h(r):= - (1+r^6)\k_{{6}}\k_{{9}} \, +\,  3 r^2(\k_{6}^2\k_{9}^2a({\eta}))^{\frac{1}{3}} (1+r^2)
	\, + \, 4 r^3 (\k_{3}\k_{6}\k_{9}\k_{12})^{\frac{1}{2}}.
	\] 
	This is a degree $6$ polynomial in $r$ with negative leading and independent term and the other coefficients are nonnegative, with at least one positive. Since 
	the right hand side of \eqref{eq:suff4} evaluated at $s=1$ is strictly positive, $h(1)>0$ and $h$ has exactly two distinct positive roots $r_1$ and $r_2$. These give rise to two values $s_1(\eta')=r_1^6, s_2(\eta')=r_2^6$, satisfying 	
	$s_1(\eta')<\tfrac{14- \sqrt{192}}{2} $ and $\tfrac{14 +  \sqrt{192}}{2}<s_2(\eta')$ for any $\eta'$, and such that \eqref{eq:suff4} holds for any $s\in [s_1(\eta'),s_2(\eta')]$. 
	This proves (i). 
	
	If $s\notin [s_1(\eta'),s_2(\eta')]$, then $h(\sqrt[6]{s})$ is negative, and hence  inequality \eqref{eq:suff4} only holds for $K_1\leq K_1'$ for $K_1'>0$ making the right-hand side of \eqref{eq:suff4} zero. This concludes the proof of (ii).
	
	Finally, (iii) follows from the fact that $a(\overline{\k})$ increases with the product $\k_3\k_{12}$, and hence the positive terms of $h(r)$ also increase. 
\end{proof}

\section{\bf Regions of Multistationarity }\label{sec:enabled}
In the previous section we gave two inequalities in the kinetic parameters that guarantee monostationarity for all choices of total amounts. Furthermore, when $K_2,K_3,\k_3,\k_6,\k_9,\k_{12}$ are fixed, Corollary~\ref{cor:apos_circuit} (see also Figure~\ref{fig:regions}) certifies monostationarity  for a cone pointed at zero and containing the line $K_1=K_4$, and leaves two regions, along the $K_1$- and $K_4$-axes, undecided. 
Now, we will show that 
if $K_4$ also is fixed, then multistationarity is enabled for $K_1$ large enough, and, symmetrically, if $K_1$ is fixed, then $K_4$ large enough yields multistationarity. 
We start by proving this fact using the Newton polytope of $p_{\eta,H}(x)$, but now viewed as a polynomial in $K_1,x_1,x_3$.
Afterwards, we give an explicit parametric description of the regions of mono- and multistationarity.

\subsection{Multistationarity can be enabled when $b(\eta)<0$. }\label{sec:multi:greenline}
Consider
 \[ \eta'=(K_2,K_3,K_4,\k_3,\k_6,\k_9,\k_{12}),\]  
and recall that we write $a(\overline{\k})=\k_3\k_{12}-\k_6\k_9$. 
Let $q_{\eta'}(K_1,x_1,x_3)$ be the polynomial $p_{\eta,H}(x_1,x_3)$ viewed as a  polynomial in $K_1,x_1,x_3$.  
Under the hypothesis $a(\overline{\k})\geq 0$ (which is independent of $K_1$), the  coefficient 
of $K_1^2  x_1^2  x_3$ is negative and equals $- K_2K_3 \k_3\k_6^2\k_9\k_{12}$. The Newton polytope of $q_{\eta'}(K_1,x_1,x_3)$ depends on whether $a(\eta)=0$ or $a(\eta)>0$, but in both cases the point $(2 ,  2 , 1)$ is a vertex.

\begin{proposition}\label{prop:multi:greenline}
Consider $\eta'=(K_2,K_3,K_4,\k_3,\k_6,\k_9,\k_{12})$ and  let $\mathcal{N}^o$  be the interior of the outer normal cone of $\newton{q_{\eta'}}$ at $(2, 2 , 1)$.
If $K_1$ belongs to the set
\[  \bigcup_{w\in \mathcal{N}^o}  \Big\{ y \st y>z_0^{w_1}, \textrm{ with }z_0\textrm{ the largest root of  }q_{\eta'}(z^{w_1},z^{w_2}, z^{w_3}) \Big\},  \] 
then $p_{\eta,H}$ attains negative values over $\R^2_{>0}$ and $\eta$ enables multistationarity. Moreover, this set is nonempty. 

Analogously, by symmetry, given $K_1,K_2,K_3,\k_3,\k_6,\k_9,\k_{12}$, after applying $\sigma$ from Remark~\ref{rk:symmetry} to $q_{\eta'}(K_1,x_1,x_3)$, we obtain a set of values of $K_4$ that enable multistationarity. 
\end{proposition} 
\begin{proof}
As   $(2 ,  2 , 1)$ is a vertex of $\newton{q_{\eta'}}$, 
 there exist $K_1,x_1,x_3>0$ such that $q_{\eta'}(K_1,x_1,x_3)<0$ by Proposition~\ref{prop:newton}.
By Remark~\ref{rk:outer_normal},  for $w\in  \mathcal{N}^o$, we consider the univariate function
$u_{\eta',w}(z)=q_{\eta'}(z^{w_1},z^{w_2}, z^{w_3})$, which is a generalized polynomial with real exponents and negative leading term. 
Then $u_{\eta',w}(z)<0$ for all $z>z_0$, where $z_0$ is the largest root of $u_{\eta',w}$. 
With $\eta = (z^{w_1},K_2,K_3,K_4,\k_3,\k_6,\k_9,\k_{12})$, we have $p_{\eta,H}(z^{w_2}, z^{w_3})= u_{\eta',w}(z)<0$. Hence, for any $K_1=z^{w_1}$ with $z>z_0$, $p_{\eta,H}$ attains negative values. 
All that remains is to show that $w_1$ is positive, to rewrite this condition as $K_1>z_0^{w_1}$ as in the statement. 

The outer normal cone $ \mathcal{N}^o$ of $\newton{q_{\eta'}}$ at $(2, 2 , 1)$  is generated by the vectors
\begin{equation}\label{eq:outervectors} 
\begin{split}
{v_1} &:= (2, 1 , 0), \quad 
{v_2} := (1,0 , 1 ), \quad
{v_3} := (2,1 , 2), \qquad \textrm{if}\quad a(\eta)>0, \\
{v_1} &:= (2, 1 , 0), \quad 
{v_2} := (1,0 , 1 ), \quad
{v_3} := (0,0 ,1), \qquad \textrm{if}\quad a(\eta)=0. 
\end{split}
\end{equation}
As any vector $\mathcal{N}^o$ is of the form 
$w=\lambda_1  {v_1}+ \lambda_2 {v_2}+ \lambda_3 {v_3}$
with $\lambda_i > 0$, we have $w_1>0$.  This concludes the proof.

Computations can be found in the supplementary file \emph{SupplInfo.mw}. 
\end{proof}

\begin{figure}[t!]
\includegraphics[scale=0.3]{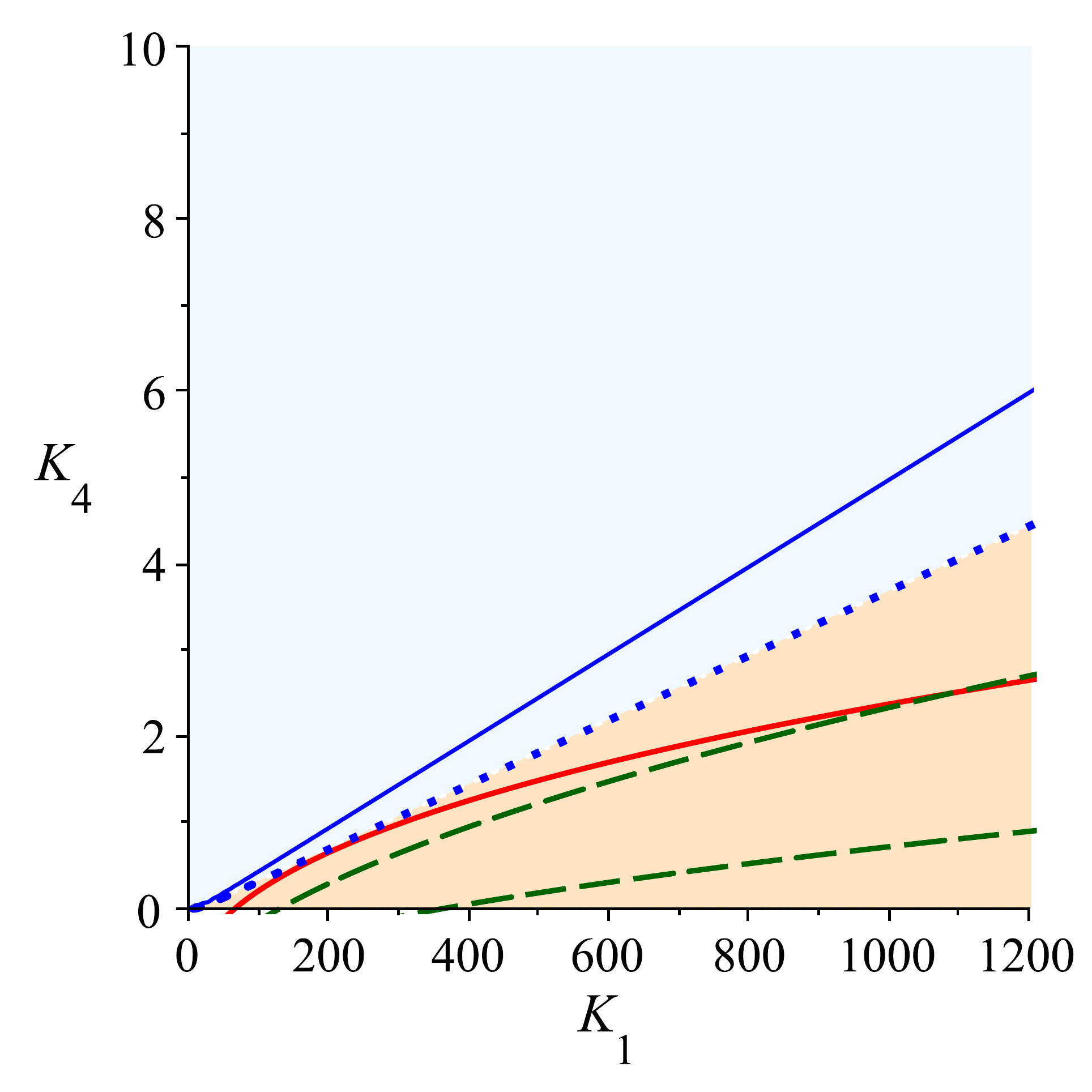}
\caption{{\small With $(K_2,K_3, \k_3,\k_6,\k_9,\k_{12})=(1,1,2,1,1,1)$, the figure shows 
a dotted blue line separating the regions of monostationary (above the line, blue) and of multistationarity (below the line, orange), found from Theorem~\ref{thm:parametric}. Above the solid blue line in the monostationarity region, the condition in Theorem~\ref{prop:apos_circuit} is satisfied. Below the solid red line in  the multistationarity region, multistationarity is enabled by means of Proposition~\ref{prop:multi:greenline} with $w=(3,1,2)\in \mathcal{N}^o$; similarly, the green dashed lines correspond to $w=(\tfrac{1}{2},2,1)$ and $w=(\tfrac{1}{3},3,2)$. 
}}\label{fig:multi:greenlines}
\end{figure}

\begin{example}\label{ex:lifting2}
Proposition~\ref{prop:multi:greenline} was invoked to select a parameter point $\eta$ enabling multistationarity in  Example ~\ref{ex:lifting}. 
Let  $\eta'=(K_2,K_3,K_4,\k_3,\k_6,\k_9,\k_{12})= (1,1,1,2,1,1,1)$, such that $a(\overline{\k})>0$. Consider the vector $(3, 1 , 2)= \tfrac{1}{2} {v_1} +  {v_2} + \tfrac{1}{2} {v_3}\in \mathcal{N}^o$ (c.f. \eqref{eq:outervectors}).
Then
\[ q_{\eta'}(z^{3},z, z^{2}) =  -z^7(-2 z^3 +11z^2 + 15z + 12), \]
whose largest root is  $\approx 6.75$.
Hence, by considering $K_1= 7^3=343$, multistationarity is enabled. 
Furthermore, this also gives that  $(x_1,x_3)= (7^{w_2}, 7^{w_3})=(7,49)$, satisfies  $p_{\eta,H}(x_1,x_3)<0$. 

\smallskip
Figure~\ref{fig:multi:greenlines} shows part of the region of Proposition~\ref{prop:multi:greenline} defined by the polynomial $q_{\eta'}(z^{3},z, z^{2}) $  (solid red line), together with the regions defined by other choices of $w$ (dashed lines in green). 
Obtaining an explicit description of the region in Proposition~\ref{prop:multi:greenline}, in terms of algebraic inequalities in the parameters has not been possible. However, in what follows we provide an explicit parametric description of the region of multistationarity (giving rise to the dotted blue line).
\end{example}

 \subsection{\textbf{Parametrization of the region of multistationarity.}}\label{sec:param}
 
Let $\eta'=(K_2,K_3,\k_3,\k_6,\k_9,\k_{12})$ for any $\eta$. Assume $a(\overline{\k})\geq 0$. 
We provide now two functions $\psi,\phi$   in $(s,\eta')$ and $(s,K_1,\eta')$ respectively, and a function $\xi(\eta')$  such that
$\eta=(K_1,K_2,K_3,K_4,\k_3,\k_6,\k_9,\k_{12})$ enables multistationarity if and only if
\begin{align}
 K_1& = \psi(s,\eta'), & K_4 & > \phi(s, \, \psi(s,\eta')\, ,\eta') , && \textrm{ for }s\in (0,\xi(\eta')),\qquad \textrm{\it or}  \label{eq:cond1} \\ 
 K_4& = \psi(s,\sigma(\eta')), &  K_1 &>\phi(s, \, \psi(s,\sigma(\eta')\, ),\sigma(\eta')), && \textrm{ for }s\in (0,\xi(\sigma(\eta')).\label{eq:cond2}
\end{align}
Note that if $K_2,K_3,\k_3,\k_6,\k_9,\k_{12}$ are fixed, then  Proposition~\ref{prop:multi:greenline} and Corollary~\ref{cor:apos_circuit}, together with the fact that $b(\eta)>0$ for $K_1,K_4$ small, indicate that there are two branches of multistationarity along the two axes: one with $K_1$ large and $K_4$ small, and one with $K_4$ large and $K_1$ small. 
These are the two branches giving rise to the two conditions \eqref{eq:cond1} and \eqref{eq:cond2}. 
By the symmetry of the system, we describe the $K_4$-branch   \eqref{eq:cond1}, and the other branch results from  applying $\sigma$. 
We specify the  nature of these branches further in the following lemma. 

\begin{lemma}\label{lemma:branches}
Assume that $\eta=(K_1^*,K_2,K_3,K_4^*,\k_3,\k_6,\k_9,\k_{12})$ enables multistationarity and $a(\overline{\k})\geq 0$. 
 Then either for all $K_4\geq K_4^*$ and $K_1\leq K_1^*$ (if $K_1^*<K_4^*$) or for all $K_4\leq K_4^*$ and $K_1\geq K_1^*$ (if $K_1^*>K_4^*$), the parameter point $\eta'=(K_1,K_2,K_3,K_4,\k_3,\k_6,\k_9,\k_{12})$ also enables multistationarity. 
\end{lemma}
\begin{proof}
As  $\eta$ enables multistationarity, there exist $x_1,x_3>0$ such that  $p_{\eta,H}(x_1,x_3)<0$. We fix these values of $x_1,x_3$, and let $\eta''=(K_1^*,K_2,K_3,\k_3,\k_6,\k_9,\k_{12})$ obtained from $\eta$ in the statement. 
The crucial observation is that $p_{\eta,H}$, with $\eta'',x_1,x_3$ fixed, is simply a 
linear polynomial $q(K_4)=c_1K_4+c_0$ in $K_4$, which satisfies $q(K_4^*)<0$. 
By Corollary~\ref{cor:apos_circuit}, $q(K_1^*)\geq 0$ (as $p_{\eta,H}(x_1,x_3)\geq 0$ if $K_4=K_1$), and hence  $c_1\neq 0$. 
If $c_1<0$, then $q(K_4)<0$ for any $K_4\geq K_4^*$, and this implies $K_4^*>K_1^*$ must hold. Similarly, if $c_1>0$, then necessarily $c_0<0$, and hence  $q(K_4)<0$ for any $K_4\leq K_4^*$, implying $K_4^*<K_1^*$. As $q(K_4)<0$ implies $\eta'$ enables multistationarity, the inequalities  in the statement regarding $K_4$ hold. The inequalities for $K_1$ follow by symmetry. 
\end{proof}

 Based on Lemma~\ref{lemma:branches}, we define the $K_4$-branch of multistationarity to consist of the set of parameters 
 $\eta=(K_1^*,K_2,K_3,K_4^*,\k_3,\k_6,\k_9,\k_{12})$ enabling multistationarity  and such that 
 $K_4^*>K_1^*$. Any point in this branch satisfies that $(K_1^*,K_2,K_3,K_4,\k_3,\k_6,\k_9,\k_{12})$ also enables multistationarity for all $K_4\geq K_4^*$. For fixed parameters $K_1^*,K_2,K_3,\k_3,\k_6,\k_9,\k_{12}$, we wish to determine the \emph{infimum} value $K_4^*$ that satisfies this property, that is, the value $K_4^*$ such that for any $K_4>K_4^*$ multistationarity is enabled. 

 In the next theorem we identify this value parametrically: we give  functions $\psi(s,\eta')$ and $\phi(s,K_1,\eta')$, for $s$ in an interval of the form $(0,\xi(\eta'))$, such that for any $K_4>\phi(s,\psi(s,\eta'),\eta')$, the point $(\psi(s,\eta'),K_2,K_3,K_4,\k_3,\k_6,\k_9,\k_{12})$ enables multistationarity, but for $K_4\leq \phi(s,\psi(s,\eta'),\eta')$, multistationarity is not enabled.
For fixed $\eta'$, the pair $(\psi(s,\eta'), \phi(s,\psi(s,\eta'),\eta'))$ describes a curve in the $(K_1,K_4)$-plane separating the region of monostationarity and multistationarity along the $K_4$-branch. 
The $K_1$-branch of multistationarity is defined analogously.

Specifically, we define the following functions in $s$, $K_1$ and $\eta'=(K_2,K_3,\k_3,\k_6,\k_9,\k_{12})\in \R^6_{>0}$: 
{\small
\begin{align*}
\alpha_{1}(s,\eta') &=- \Big( K_{2} ( K_{2}+K_{3} )\k_{3}\k_{9}\k_{12}\, s
+K_{3}\k_{12} ( 2\,K_{2}a(\overline{\k})+(K_2+K_{3})\k_{3}\k_{12} ) \\ & \quad +\sqrt{K_{2}K_{3}\k_{3}
 \k_{12}a(\overline{\k})} ( 2K_{2}\k_{9}\, s+K_{2}\k_{12}+3\,K_{3}\k_{12} )  \Big) K_{2}\k_{3}^{2}s^{3}, 
\\
\beta_{1}(s,\eta') &= \k_6\Big( -K_{2}^{2}\k_{3}^{2}\k_{9}^{2}\, s^{4}+K_{2}\k_{3}^{2}\k_{9}\k_{12} ( 3K_{2}-K_{3} )\, s^{3}+2K_{2}K_{3}\k_{3}\k_{12} ( 4\k_{3}\k_{12}-\k_{9}\k_{6} )\,  s^{2}
\\ & \quad -K_{3}\k_{3}\k_{6}\k_{12}^{2} ( K_{2}-3K_{3} ) \, s -K_{3}^{2}\k_{6}^{2}\k_{12}^{2} \\ & \quad +2\,\sqrt{K_{2}K_{3}\k_{3}\k_{12}a(\overline{\k})}s  \big( K_{2}\k_{3}\k_{9}\, s^{2}+ 2( K_{2}+K_{3})\k_{3}\k_{12}\, s+K_{3}\k_{6}\k_{12} \big)  \Big),
\end{align*}}
 and
 {\small 
\begin{align*}
\alpha_{4}(s,K_1,\eta') &= K_3\k_{12}   \Big( -2\,\sqrt{K_2K_3\k_3\k_{12} a(\overline{\k})}( K_{{2}}\k_{{3}}s+K_1\k_6) s  + 
 \\ & \qquad 
K_2\k_3 ( K_1 \k_6\k_9-(K_2+K_3)\k_3\k_{12}) s^2 - 2 K_1K_2\k_3\k_6\k_{12} s - K_1K_3 \k_6^2 \k_{12}\Big),      \\
\beta_{4}(s,\eta') &= K_2 \k_3\k_9 s^2 \Big(2\sqrt {K_{2}K_{3}\k_{3} \k_{12} a(\overline{\k}) } \, s +\, ( K_{2}\k_{3}\k_{9}\, s^{2}+2K_{3}\k_{3}\k_{12}\, s -K_{3}\k_{6}\k_{12} ) \Big).
\end{align*}}%
We let now 
\begin{align}
\psi (s,\eta')= \frac{\alpha_{1}(s,\eta')}{\beta_{1}(s,\eta') },\qquad \phi(s,K_1,\eta')= \frac{\alpha_{4}(s,K_1,\eta')}{\beta_{4}(s,\eta') }, \label{eq:paramphi}
\end{align}
and let $\xi(\eta')$ be the first positive root of the polynomial $\beta_{1}(s,\eta')$ with variable $s$ and $\eta'$ fixed.

\begin{theorem}\label{thm:parametric}
Let $\eta=(K_1,K_2,K_3,K_4,\k_3,\k_6,\k_9,\k_{12})\in \R^8_{>0}$  such that $\k_3\k_{12}-\k_6\k_{12}\geq 0$, and denote $\eta'=(K_2,K_3,\k_3,\k_6,\k_9,\k_{12})$. Recall the map $\sigma$ from Remark~\ref{rk:symmetry}.
Multistationarity is enabled if and only if $K_1,K_4$ are as in one of the following cases:
 \[ K_1 = \psi(s,\eta'),\quad\textrm{and}\quad K_4 > \phi(s, \, \psi(s,\eta')\, ,\eta'), \qquad\textrm{with }\quad s\in (0,\xi(\eta')).\]
or 
\[ K_4 =\psi(s,\sigma(\eta')),\quad\textrm{and}\quad K_1 >\phi(s, \, \psi(s,\sigma(\eta')\, ),\sigma(\eta')), \qquad\textrm{with }\quad s\in (0,\xi(\sigma(\eta')).\]
The first case describes the $K_4$-branch, and $s=x_1$, while the second case describes the $K_1$-branch, and $s=x_2$. 
Furthermore, for any $\eta'$, $\psi$ increases for $s$ in the considered interval and the image is $\R_{>0}$. 
\end{theorem}
\begin{proof}
We consider $\eta'$ fixed and study the $K_4$-branch. The proof relies on several symbolic computations that can be found in the accompanying supplementary file \emph{SupplInfo.mw}. 
 Recall from the proof of Lemma~\ref{lemma:branches} that $p_{\eta,H}(x_1,x_3)$ is linear in $K_4$. If written as $c_1 K_4 + c_0$ we have 
\begin{align*}
c_1 &=K_2\k_3\k_9 x_1^2 \Big(K_2\k_3 a(\overline{\k})x_1^{2} x_3^{2}+K_1\k_6 \big( K_2\k_3\k_9x_1^{2}+2\,K_3\k_3\k_{12}x_1-K_3\k_6\k_{12}\big) x_3+K_1^{2}K_3\k_6^{2}\k_{12}\Big) \\
c_0 &= K_3\k_{12} \Big( K_{2}\k_{3}a(\overline{\k}) ( K_{2}\k_{3}x_{1}+K_{1}\k_{6} )x_{1}^{2}x_{3}^{2}-K_{1}\k_{6} \big( 
K_{2}\k_3( K_{1}\k_{6}\k_{9} -(K_{2}+K_{3})\k_{3}\k_{12} ) x_{1}^{2} \\ & 
- K_1 \k_{6}\k_{12}(2K_{2}\k_{3}x_{1}+ K_{3}\k_{6}) \big) x_{3}+K_{1}^{2}K_{3}\k_{6}^{2}\k_{12} (K_{2}\k_{3} x_{1}+K_{1}\k_{6} ) 
 \Big).
\end{align*}
In order to understand the $K_4$-branch, we consider the case $c_1<0$ (see the proof of Lemma~\ref{lemma:branches}). 
For fixed $x_1,x_3,K_1$, this implies that the coefficient of $x_3$ in $c_1$  is negative, which in turn implies that $x_1$ is smaller than the positive root of $K_2\k_3\k_9x_1^{2}+2K_3\k_3\k_{12}x_1-K_3\k_6\k_{12}$, namely, smaller than
\[ x_{1,{\rm bound}} := \tfrac{-K_3\k_3\k_{12}+\sqrt {K_{3}\k_{{3}}\k_{{12}} ( K_{2}\k_6\k_9+K_{3}\k_3\k_{12}) }}{K_{{2}}\k_{{3}}\k_{{9}}}.\]
Under the assumption $x_1<x_{1,{\rm bound}}$, and $a(\overline{\k})\geq 0$, using the function {\tt IsEmpty} in {\tt Maple 2019}, we find that  $c_0>0$. Hence for $\eta$ in the $K_4$-branch, if $p_{\eta,H}(x_1,x_3)<0$, then necessarily $c_1<0$ and $c_0>0$.
Furthermore, in this case $p_{\eta,H}(x_1,x_3)=0$ holds if and only if $K_4=\tfrac{-c_0}{c_1}>0$, and $p_{\eta,H}(x_1,x_3)<0$ holds if $K_4>\tfrac{-c_0}{c_1}$.
It follows that the boundary of the $K_4$-branch is determined by minimizing $\tfrac{-c_0}{c_1}$ with respect to  $x_1,x_3>0$ subject to $c_1<0$.  For $a(\overline{\k})>0$, we find the minimum value of  $\tfrac{-c_0}{c_1}$, and for $a(\overline{\k})=0$, we find its infimum value.

For a fixed $x_1>0$, we consider first $\tfrac{-c_0}{c_1}$ as a function of $x_3$ in the region where $c_1<0$. 
When $a(\overline{\k})>0$, the derivative has a unique positive zero at
\[ x_{3,{\rm min}} := \tfrac{K_1\k_6 \sqrt{K_2K_3\k_3\k_{12} a(\overline{\k})}}{K_2\k_3a(\overline{\k})x_1},\]
which defines a minimum. 
We evaluate $\tfrac{-c_0}{c_1}$ at $x_{3,{\rm min}}$, which now becomes the function $\phi(x_1,K_1,\eta')$ in \eqref{eq:paramphi}.
When $a(\overline{\k})=0$, $\tfrac{-c_0}{c_1}$ is strictly decreasing, and hence the infimum value it attains is the limit as $x_3$ goes to $+\infty$, which is $\phi(x_1,K_1,\eta')$ again. It makes sense then to set $x_{3,{\rm min}}=+\infty$ in this case.  
Hence $\phi(x_1,K_1,\eta')$ gives, for fixed $\eta'$, $K_1$, and $x_1$ such that $c_1<0$, the minimal/infimum value of $\tfrac{-c_0}{c_1}$ seen as a function of $x_3$. 

We notice that the denominator of  $\phi(x_1,K_1,\eta')$ (which is a multiple of $c_1(x_1,x_{3,{\rm min}})$ when $a(\overline{\k})>0$), is a polynomial in $x_1$ of the form  $x_1^2$ times a quadratic polynomial. The latter has positive leading term and negative independent term. Hence it has a unique positive root $\gamma$ (which we can compute), and this denominator is negative  if and only if $x_1\in (0,\gamma)$.  When $a(\overline{\k})=0$, we have  $\gamma=x_{1,{\rm bound}}$.

In particular  $\phi$ is continuous and differentiable in $(0,\gamma)$. The function $\phi$ is a rational function in $x_1$ of the following form:
\begin{equation*}
\phi(x_1,K_1,\eta')= \tfrac{a_1 x_1^2 - a_2 x_1 - a_3}{x_1^2 ( b_1x_1^2 + b_2 x_1 - b_3)},
\end{equation*}
where $a_2,a_3,b_1,b_2,b_3$ depend on $\eta',K_1$ and are positive under the current hypotheses, and $a_1$, which also depends on $K_1,\eta'$ is
\[ 
a_1 := - K_2K_{{3}}\k_3\k_{{12}} \Big( -K_1\k_6\k_9 + (K_{{2}}+K_{{3}})\k_{{3}}\k_{{12}})
+ 2\sqrt {K_2 K_3 \k_3\k_{12}a(\overline{\k})}\, \Big).
\]
In order to minimize $\phi$ in $(0,\gamma)$, we find the derivative of $\phi$ with respect to $x_1$:
\[
\phi'(x_1,K_1,\eta')= \tfrac{-2a_1b_1x_1^4 + (3a_2b_1-a_1b_2) x_1^3 + (2a_2b_2 + 4a_3b_1)x_1^2 + (3a_3b_2 - a_2b_3) x_1 - 2a_3b_3}{x_1^3 ( b_1x_1^2 + b_2 x_1 - b_3)^2}.
\]
The extreme values of $\phi'$ are determined by the zeroes of its numerator. 
This numerator is a polynomial $u(x_1)$ in $x_1$ with negative independent term and positive degree $2$ term. 
If $a_1\leq 0$, then the leading and degree $3$ coefficients of $u(x_1)$ are nonnegative. By Descartes' rule of signs, it follows that $\phi'=0$ has exactly one positive root, which, in case it belongs to $(0,\gamma)$, gives  rise to a minimum of $\phi$, as the independent term of the numerator of $\phi'$ is 
negative. 

If $a_1>0$,  then the leading term of $u(x_1)$ is negative, and by the Descartes' rule of signs,  $\phi'=0$ at most two positive roots, in which case the first positive root will be a minimum of $\phi$ if it belongs to $(0,\gamma)$ as above. 
Note that $a_1>0$ if and only if 
\begin{equation*} 
K_1 > K_{1,{\rm bound}},\qquad \textrm{where}\quad  K_{1,{\rm bound}}:= \tfrac{(K_2+K_3)\k_3\k_{12} + 2 \sqrt{K_2 K_3 \k_3\k_{12}a(\overline{\k}) }  }{\k_6\k_9}.
\end{equation*}

\begin{figure}[t!]
\begin{tikzpicture}[scale=0.8]
\draw[-] (-1,0) -- (7,0);
\draw[-] (0,-1) -- (0,6);
\draw[-,dashed,gray] (2,-1) -- (2,6);
\draw[-,dashed,gray] (4,-1) -- (4,6);
\node (a) at (1.6,-0.3) {{\scriptsize $\xi(\eta')$}};
\node (a) at (3,-0.3) {{\scriptsize $\gamma$}};
\draw[fill] (3,0) circle (2pt);
\draw[fill] (2,0) circle (2pt);
\draw[fill] (4,0) circle (2pt);
\node (a) at (4.5,-0.3) {{\scriptsize $\overline{\xi}(\eta')$}};
\node (a) at (7.4,0) {{\small $s$}};
\node (a) at (-0.3,6) {{\small $\psi$}};
\node (a) at (-0.8,2.2) {{\scriptsize $K_{1,{\rm bound}}$}};
\draw[-,dashed,red] (0,2.2) -- (7,2.2);
\draw[-,blue, line width=0.8pt] (0,0) .. controls (1,0.3) and (2,0.5) .. (1.9,6);
\draw[-,blue, line width=0.8pt] (4.1,6) .. controls (4,3) and (4.5,2.5) .. (7,2.3);
\end{tikzpicture}
\caption{{\small Cartoon depiction of the function $\psi(s,\eta')$ for a fixed $\eta'$, with $\xi(\eta'), \overline{\xi}(\eta')$, $\gamma$ and $K_{1,{\rm bound}}$ as given in the proof of Theorem~\ref{thm:parametric}.}} \label{fig:psi}
\end{figure}
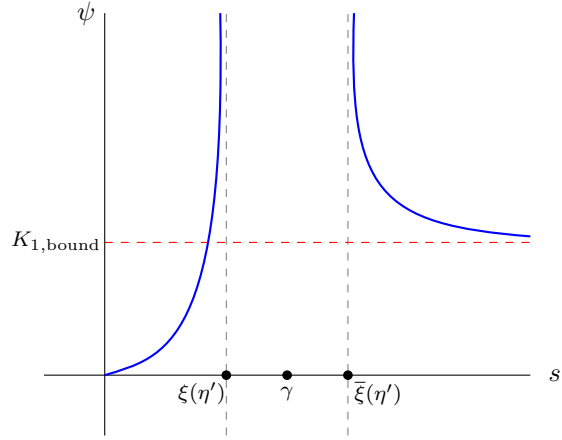

The next step is thus to confirm that the only positive root in the case $a_1\leq 0$ is smaller than $\gamma$, and that there is such a (simple) positive root in the case $a_1>0$.
To this end, we observe that the numerator of $\phi'$ is linear in $K_1$. By solving the numerator for $K_1$, we obtain that any extreme value satisfies
\[ K_1 = \psi(x_1,\eta'),\]
with $\psi$ as in \eqref{eq:paramphi}.
The denominator $\beta_1(x_1,\eta')$  has degree $4$ in $x_1$, negative leading and independent terms, and the coefficient of $x_1^2$ is positive. By Descartes' rule of signs, $\beta_1(x_1,\eta')$ has at most two positive roots. 
Using the function {\tt IsEmpty} in {\tt Maple 2019}, we find that $\beta_1(\gamma,\eta')>0$. This implies that $\beta_{1}(x_1,\eta')$ has exactly one simple positive root $\xi(\eta')$ in the interval $(0,\gamma)$ and one simple positive root $\overline{\xi}(\eta')$ in $(\gamma,+\infty)$. 
The numerator $\alpha_1(x_1,\eta')$ of $\psi$ has degree $4$ in $x_1$, is negative for  $x_1>0$, and vanishes at $x_1=0$. 
Hence,  $\psi(x_1,\eta')$ is positive in the intervals $(0,\xi(\eta'))$ and $(\overline{\xi}(\eta'),+\infty)$. 
It tends to infinity when $x_1$ tends to $\xi(\eta')$ from the left and also   to $\overline{\xi}(\eta')$ from the right. Furthermore, $\psi$ vanishes at $x_1=0$ and tends to $K_{1,{\rm bound}}$ when $x_1$ tends to infinity. In particular, the image of $\psi$ over the interval $(0,\xi(\eta'))$ is $\R_{>0}$, and the image over the interval  $(\overline{\xi}(\eta'),+\infty)$ is $(K_{1,{\rm bound}},+\infty)$. See Figure \ref{fig:psi}.  The image  of  $(\xi(\eta'),\overline{\xi}(\eta'))$ by $\psi$ belongs to $\R_{<0}$. 

The anti-images of a given $K_1$ by $\psi$ are the zeroes of $\phi'=0$. By comparing the image of $\psi$ to the discussion on the sign of $a_1$ and the positive roots of $\phi'$ above, we conclude that  $\psi$ is strictly increasing in  $(0,\xi(\eta'))$, and each  $x_1$ in this interval such that  $K_1=\psi(x_1,\eta')$ is a simple root of $\phi'=0$. 
In particular, $\phi$ attains its minimum at the anti-image of $K_1$ by $\psi$ in the  interval $(0,\gamma)$. 

To summarize, we have shown that given $K_1>0$, and $\overline{x}_1\in (0,\xi(\eta'))$ such that $K_1=\psi(\overline{x}_1,\eta')$, 
$K_4$ gives rise to a parameter point enabling multistationarity in the $K_4$-branch if and only if $K_4$ is larger than $\tfrac{-c_0}{c_1}$ evaluated at $x_{3,{\rm min}}$  and $\overline{x}_1$, where we already know that $c_1<0$ as $\xi(\eta')<\gamma$. 
This gives that $\eta$ enables multistationarity in the $K_4$-branch if and only if there exists $x_1\in (0,\xi(\eta'))$ such that $K_1=\psi(x_1,\eta')$ and $K_4>\phi (x_1,\psi(x_1,\eta'),\eta')$. 
This concludes the proof of (i); (ii) follows by symmetry using Remark~\ref{rk:symmetry}.
\end{proof}

Figure~\ref{fig:multi:greenlines} shows the $K_1$-branch of the multistationarity region given in Theorem~\ref{thm:parametric} when $(K_2,K_3,\k_3,\k_6,\k_9,\k_{12})= (1,1,2,1,1,1)$. 
 
Using implicitacion via for example Gr\"obner bases, one could theoretically determine an implicit equation for the curve 
$(\psi(s,\eta'), \phi(s,\psi(s,\eta'),\eta'))$ in the $(K_1,K_4)$-plane for a fixed $\eta'$. 
Such a computation has not been possible for arbitrary $\eta'$ due to the computational cost. For $\eta'$ fixed, as in Figure~\ref{fig:multi:greenlines}, we obtain a polynomial in $K_1,K_4$ whose zero set includes the dotted blue curve   in Figure~\ref{fig:multi:greenlines} given by the parametrization, as well as additional components. 

\begin{remark}
Theorem ~\ref{thm:parametric} provide a means to verify whether a given $\eta$ enables multistationarity: First, decide  whether Theorem~\ref{prop:apos_circuit} is informative. If not, and $K_4>K_1$, then determine $s\in (0,\xi(\eta'))$ such that $K_1 =\psi(s,\eta')$ for $s\in (0,\xi(\eta'))$, and decide whether $K_4 >\phi(s,\psi(s,\eta'),\eta')$. If $K_1>K_4$, use the expressions for the $K_1$-branch. 

For example, let $\eta= (3,1,1,700,2,1,1,1)$.  Inequality \eqref{Inequality:SufficientInequality3} in Theorem~\ref{prop:apos_circuit} does not hold. As $K_4>K_1$, we consider the $K_4$-branch. 
We solve 
$3= \psi(s,\eta')$ for $s\in (0,\xi(\eta'))$ and obtain $s\approx 0.174$, which gives $\phi(s,\psi(s,\eta'),\eta') \approx 818.17$. 
As $700<818.17$, the given parameter point does not enable multistationarity. 
It follows as well that the parameter point $ (3,1,1,900,2,1,1,1)$ enables multistationarity. 
\end{remark}

\section{\bf Connectivity}\label{sec:connected}

In this section we show that the open set $X\subseteq \R_{> 0}^8$  of parameter points that enable multistationarity is  connected.
As any $\eta\in \R^8_{>0}$ either enables or precludes multistationarity, the set $\R^8_{>0}\setminus X$ consists of the parameter points that preclude multistationarity.  

We consider $X$ as a topological subspace of $\R^8_{>0}$ with the Euclidean topology. We start by highlighting in the next lemma a path connected subset of $X$. Let $Y\subseteq \R_{> 0}^8$ consist of the parameter points $\eta$ such that $a(\eta)<0$.

\begin{lemma}\label{lem:connect}
The following subsets of $\R^4$ are path connected:
\[ A_{<0}=\{ \overline{\k}=(\k_3,\k_6,\k_9,\k_{12})\in \R^4_{>0} \st a(\overline{\k})<0\},\quad  A_{\geq 0}=\{ \overline{\k}=(\k_3,\k_6,\k_9,\k_{12})\in \R^4_{>0} \st a(\overline{\k})\geq 0\}. \]
Additionally, 
$Y$ is path connected.
\end{lemma}

\begin{proof}
Consider the continuous map $h\colon \R^4_{>0}\rightarrow \R^2_{>0}$ sending $(\k_3,\k_6,\k_9,\k_{12})$
to $(\k_3\k_{12},\k_6\k_9)$. The fibers of this map are path connected.
As $A_{<0}$ and $A_{\geq 0}$ are respectively the preimages by $h$ of the path connected subsets $\{ x\in \R^2_{>0} \st x_1<x_2\}$ and $\{ x\in \R^2_{>0} \st x_1\geq x_2\}$ of $\R^2_{>0}$, they are also path connected.
$Y$ is also path connected as it is homeomorphic to $\R^4_{>0} \times A_{<0}$.
\end{proof}

By Proposition~\ref{prop:summary}, multistationarity is enabled whenever $a(\eta)<0$. Therefore, $Y_{<0}$  is a subset of $X$. To show that $X$ is path connected it is enough to show that there exists a path from any point  in $X$  to a point in $Y_{<0}.$

\begin{theorem}	\label{thm:connected}
$X$ and $\R^8_{>0}\setminus X$  are path connected.
\end{theorem}

 \begin{proof}
 We start by showing that $X$ is path connected. 
	Let $\eta =(K_1,K_2,K_3,K_4,\k_3,\k_{6},\k_{9},\k_{12}')\in X$ such that $a(\eta)\geq 0$. 
By Lemma~\ref{lem:connect}, it is enough to show that there exists a path in $X$ that connects $\eta$ to a point $\eta^*\in Y_{<0}$. 
	As $\eta\in X$ and $a(\eta)\geq 0$, we can choose $z_1,z_3 >0$ such that $p_{\eta,H}(z_1,z_3)<0$ (c.f. Proposition~\ref{prop:summary}).  
We let $\eta'=(K_1,K_2,K_3,K_4,\k_3,\k_{6},\k_{9})$ and let $\overline{p}_{\eta',H}(x_1,x_3,\k_{12})$ denote $p_{\eta,H}$ seen as a polynomial in $x_1,x_3,\k_{12}$. 
The vertices of the Newton polytope of $\overline{p}_{\eta',H}$ are (c.f. Figure~\ref{fig:connected}):
$\ \big\{  (0, 1, 2 ),	(2, 2, 1 ), 	(2, 2, 2 ),	(1, 0, 2 ),	(2, 0, 1 ),	(0, 0, 2 ), 	(3, 2, 2 ),$ $	(4, 2, 1 ),
	(4, 1, 0 ),
	(4, 2, 0 )
	\big\}.
$
The coefficients of the vertices $(2, 2, 1)$ and $(4, 2, 0)$ are negative. These two vertices lie on the one dimensional face $F$ given by the intersection of the supporting hyperplanes $x_3-2=0$ and $-x_1-2 \k_{12}+4=0$. Therefore, the outer normal cone at $F$ is generated by the vectors $v_1:=(0,1,0)$ and $v_2:=(-1,0,-2).$ Following Remark~\ref{rk:outer_normal}, we consider  $w:=v_1+v_2=(-1,1,-2)$ 
and evaluate $\overline{p}_{\eta',H}$ at $(z_1 s^{-1},z_3 s, \k'_{12} s^{-2})$. The denominator is positive  and the numerator is
{\small \begin{align*}
\begin{split}
q(s):=& - K_{2}\k_{3} \k_{6} \k_{9}  z_{1}^{2} z_{3}^{2}(K_{2} K_{4} \k_{3} \k_{9} z_{1}^{2}  + K_{1}  K_{3} \k_{6} \k_{12}' ) s^3 +
\k_{6} z_{3} ( K_{2}^2  \k_{3}^2  \k_{9} (K_{1}  K_{4}   \k_{9} z_{1}^4 
-   K_{3}  \k'_{12} z_{1}^3 z_{3} )\\ & - K_{1} K_{2} K_{3} \k_{3} \k_{6} \k_{9} \k'_{12}(K_{1}+ K_{4})  z_{1}^2 
+ K_{1}^2 K_{3}^2 \k_{6}^2 \k_{12}^{\prime 2} ) s^2
 + \big(  K_{2} K_{4}\k_{3} \k_{9} \k'_{12}  z_{1}^2 (K_{2} \k_{3}^2  z_{1}^2 z_{3}^2\\ & + 2 K_{1}  K_{3}\k_{3} \k_{6}  z_{1} z_{3} + K_{1}^2 K_{3}  \k_{6}^2 ) 
+ K_{1}  K_{3}  \k_6\k_{12}^{\prime 2} (K_{2} \k_{3}^2   z_{1}^2 z_{3}^2 + 2 K_{1} K_{2} \k_{3} \k_{6}  z_{1} z_{3} + K_{1}^2 K_{3} \k_{6}^2)\big) s\\ &
+ K_2K_3\k_3 \k_{12}^{\prime 2} z_1 \Big(K_{2} \k_{3}^2   z_{1}^2 z_{3}^2+ K_{1} \k_{3} \k_{6} ( K_{2} +K_3) z_{1} z_{3} 
 + K_{1}^2  K_{3} \k_{6}^2\Big).
\end{split}
\end{align*}}
\begin{figure}[t!]
\includegraphics[scale=0.3]{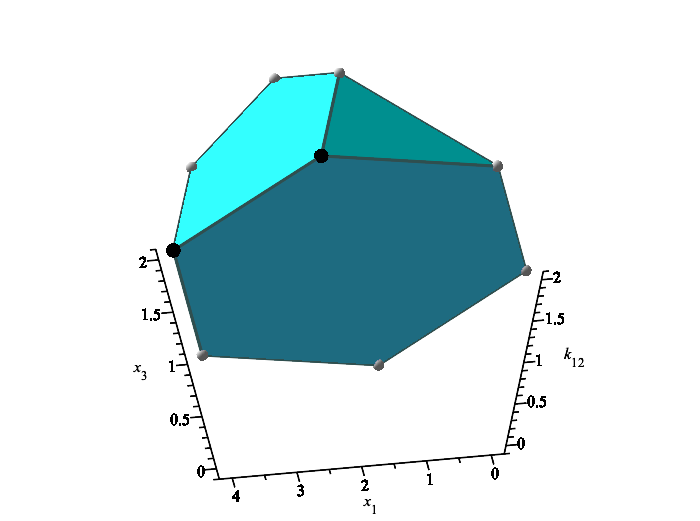}
\caption{{\small Newton Polytope of $\overline{p}_{\eta',H}$ as a polynomial in $x_1,x_3,\k_{12}$. In black we show two negative vertices. }} \label{fig:connected}
\end{figure}

The polynomial $q$  has degree 3 in $s$, its leading coefficient is negative and the coefficients of degree $0$ and $1$ are positive. By Descartes' rule of signs, $q$ has exactly one positive root. For $s=1$,  $q(1)= \overline{p}_{\eta,H}(z_1,z_3,\k'_{12})$ is negative, from where it follows that $q(s)<0$ for all $s\geq 1.$  Hence,   
$\eta(s) =(K_1,K_2,K_3,K_4,\k_3,\k_{6},\k_{9},\k'_{12}s^{-2})\in X$ for all $s\geq 1$.
As $s$ increases, $\k'_{12} s^{-2}$ decreases and hence $a(\eta(s))$  decreases. For  $s>\sqrt{\frac{\k_3 \k'_{12}}{\k_6 \k_9}}$, we have $a(\eta(s))<0$  and hence $a(\eta)\in Y_{<0}$. This provides the desired path, which proves 
 the first part of the statement.

\smallskip
To study $Z:=\R^8_{>0}\setminus X$, note that the set of points $\eta$ with $K_1=K_4$ and $a(\overline{\k})\geq 0$ is path connected by Lemma~\ref{lem:connect}, and is  further a subset of $Z$ by Corollary~\ref{cor:apos_circuit}. 
By Lemma~\ref{lemma:branches}, in $Z$ there are paths joining any $\eta=(K_1,K_2,K_3,K_4,\k_3,\k_6,\k_9,\k_{12})$ in $Z$ to  $\eta'=(K_1,K_2,K_3,K_1,\k_3,\k_6,\k_9,\k_{12})$. Hence $Z:=\R^8_{>0}\setminus X$ is path connected. This concludes the proof of the theorem.
 \end{proof}

\begin{remark}
According to Theorem~\ref{thm:connected}, the region $X$ of parameters $\eta$ that enable multistationarity is connected in $\R_{>0}^8.$ For this system, the preimage of $X$ by $\pi$, that is, the set of parameters $\k\in \R^{12}_{>0}$ that enable multistationarity, is also path connected in $\R_{>0}^{12}$. To see this, it is enough to study the map $(\k_1,\k_2,\k_3) \mapsto (\frac{\k_2 +\k_3}{\k_1},\k_3)$. The fiber of this map of each point in the image is one dimensional 
and connected. The map $\pi$ comprises four disjoint copies of such a map, and hence the fiber by $\pi$ of a point in the image is four dimensional and connected. Therefore, the preimage of $X$ by $\pi$ is path-connected. 
\end{remark}

\subsection*{Acknowledgements} EF and NK acknowledge funding from the Independent Research Fund of Denmark. The project was started while NK was at MPI, MIS Leipzig and further developed while NK was at the University of Copenhagen. TdW and OY acknowledge the funding from the DFG grant WO 2206/1-1. Bernd Sturmfels is gratefully acknowledged for useful discussions and for bringing the authors together. Alicia Dickenstein and Carsten Wiuf are thanked for comments on the manuscript.


 \end{document}